\documentclass[11pt]{article}

\usepackage[english]{babel}
\usepackage{graphicx}
\usepackage{subcaption}
\usepackage{amssymb}  
\usepackage{bbm}
\usepackage{ragged2e}
\usepackage[letterpaper,top=2cm,bottom=2cm,left=3cm,right=3cm,marginparwidth=1.75cm]{geometry}

\usepackage{bbm}
\usepackage{amsfonts}
\usepackage{graphicx}
\usepackage[colorlinks=true, allcolors=blue]{hyperref}
\usepackage{booktabs}
\usepackage{makecell}
\usepackage{setspace}
\usepackage{float}
\usepackage[authoryear]{natbib}
\usepackage{amsmath,amssymb,amsthm,mathtools}
\newtheorem{lemma}{Lemma}

\newcommand{\pr}{\Pr}
\newcommand{\be}{\begin{equation*}\begin{aligned}}  
\newcommand{\ee}{\end{aligned}\end{equation*}}

\title{The Effect of Age at Arrival on the Alignment Between Immigrant and Native-Born Gender Norms: A Distributional Approach
}
\author{Nadav Kunievsky\footnote{Knowledge Lab, University of Chicago. I thank Natalia Goldshtein for extremely helpful discussions.}}
\begin{document}
\maketitle

\begin{abstract}
This paper examines how age at migration affects cultural assimilation by studying convergence in gender role attitudes between immigrants and the UK-born population. Although cultural values are central to policy debates about integration and social cohesion, most work on migration timing focuses on economic outcomes, leaving effects on values and beliefs far less explored. We address this gap by combining a sibling design with a distributional framework for measuring attitude convergence. Using the UK Household Longitudinal Study, we compare siblings within the same family who arrived in the UK at different ages, exploiting within-family variation to identify the causal effect of childhood exposure to host-country norms. To measure convergence, we compare the full distributions of ordinal survey responses to questions on gender norms for immigrants and locals. Our distance metric is the total variation distance between response distributions. TV has a clear policy-relevant interpretation: it equals the worst-case difference in mean responses over all bounded scoring rules. We then use our estimates to construct two measures of how migration timing changes this distance. The first asks how large the immigrant–UK-born TV distance would be if every immigrant had arrived at birth, and compares it to the observed distance. The second is a marginal measure that asks how the distance changes under a small uniform shift in arrival ages. Our results show that if all immigrants had arrived at birth, the cultural distance between immigrants and locals would decrease substantially, and that marginal increases in migration age incrementally widen this gap. These effects are considerably stronger for immigrants from culturally distant origins. Overall, the findings highlight the importance of early-life exposure in shaping cultural beliefs and provide a robust, broadly applicable framework for quantifying convergence in survey responses.

\end{abstract}

\newpage

\section{Introduction}
Migration has long been recognized as a transformative process, reshaping not only immigrants' economic prospects but also their cultural identities and social attitudes. A well-established body of research documents how age at migration affects economic outcomes such as education, employment, and earnings, consistently finding that younger immigrants integrate more successfully into labor markets \citep{BleakleyChin2004,VanOursVeenman2006}. However, much less is known about how migration—particularly the age at which it occurs—affects more subjective aspects of integration, such as values, attitudes, and cultural beliefs.

This gap is especially salient given the central role of values in shaping social cohesion and integration. Policymakers frequently emphasize the alignment of immigrants' beliefs with those of the host society, whether through citizenship tests, integration programs, or political debates about shared values. Questions about gender equality, tolerance, and civic principles are often foregrounded as markers of successful assimilation. Despite this policy relevance, the literature has largely overlooked how age at migration shapes cultural assimilation, and in particular, attitudes toward gender roles.
This paper addresses this gap by examining the causal effect of migration age on gender role attitudes among immigrants in the United Kingdom (UK). 

Studying the effect of migration age on cultural attitudes presents a distinct measurement challenge compared to economic outcomes. Unlike earnings or education—where having more is almost universally considered better—there is no universally optimal attitude toward gender roles. Different policymakers may hold different views about what constitutes ideal beliefs and how differences between groups should be weighted. Moreover, in many cases what matters for social cohesion is not the absolute level of adherence to particular cultural norms, but rather the degree of alignment between immigrant and local attitudes. We therefore need measures that can accommodate disagreement about ideal attitudes and quantify cultural distance in a policy-relevant way.

In this paper we develop a distributional framework for measuring convergence that avoids imposing arbitrary numerical values on survey responses. Rather than comparing mean responses across groups—which requires assuming specific weights for each response category—we adopt the Total Variation (TV) distance between response distributions. This measure has two complementary interpretations that make it particularly valuable for policy analysis of cultural differences. First, it represents the maximum possible divergence in attitudes across all valid scoring functions, providing a worst-case average difference that is robust to how responses are coded. Second, it quantifies the minimal share of immigrants who would need to change their responses for their distribution to match that of locals, offering an intuitive coupling interpretation.

Motivated by these interpretations, we introduce two causal measures of how migration age affects cultural convergence. The first evaluates a counterfactual in which all immigrants are assumed to have arrived at age zero—effectively born in the UK. This measure captures how much closer immigrants' attitudes would be to locals' if all migration occurred before birth, providing a conceptual benchmark for the total effect of childhood exposure. The second measure, the Marginal Total Variation Divergence (MTVD), captures the marginal effect of migration age on cultural distance. It quantifies how a small uniform increase in migration age shifts the similarity between immigrants' and locals' attitudes, offering a policy-relevant metric for evaluating interventions that might induce earlier migration.

To identify the causal effect of migration age, we exploit within-family variation by comparing siblings who migrated at different ages. This sibling fixed-effects design controls for unobserved family-level heterogeneity—including parental values, origin-country context, and the circumstances surrounding migration—that jointly shape both migration timing and gender attitudes. The key identifying assumption is that, conditional on family membership, differences in siblings' arrival ages primarily reflect birth timing relative to the family's migration date rather than child-specific characteristics that independently affect attitudes. 

We focus on childhood migration because children arrive during formative stages of socialization and identity development, making them particularly susceptible to host-country cultural norms. Immigrants who arrive at younger ages may internalize the gender norms of the host society more deeply than those who arrive later in life, when attitudes are already more firmly established \citep{Rumbaut2004}.

We apply this framework using data from the UK Household Longitudinal Study (UKHLS), which provides information on migration history, family relationships, and attitudes toward gender roles. Our empirical analysis examines responses to four questions capturing beliefs about working mothers, gender-based household specialization, and dual income responsibilities. To examine heterogeneity by cultural origin, we combine UKHLS data with information from the World Values Survey, constructing measures of origin-country gender norms and their distance from UK norms.

Applying this sibling fixed-effects design and our distributional measures, we find that earlier arrival leads to substantial convergence in gender-role attitudes toward the UK-born distribution, with larger effects for immigrants from culturally distant origins. We start by documenting the size of the initial disagreement between immigrants and UK-born locals. First, the TV distance between immigrants and locals ranges from 0.047 to 0.148 across the four gender attitude questions. This means that between 5\% and 15\% of immigrants would need to change their responses for the distributions to align perfectly with locals. Equivalently, this distance captures the greatest gap in attitudes that could be obtained by assigning any bounded set of weights to the response categories and comparing the resulting average score across immigrants and UK-born locals. Because it takes the worst case over all such weighting schemes, it does not depend on any particular normative choice about how to value agreement relative to disagreement. For policymakers, this is useful because it provides an upper bound on how far the two groups could differ under any reasonable policy-relevant scoring of the responses, so conclusions about “how different” attitudes are do not hinge on adopting a specific (and potentially contested) weighting scheme chosen by researchers.

Second, we estimate that if all immigrants in our sample had been born in the UK (i.e., setting migration age to zero), the TV distance would decline substantially. For the statement "A husband's job is to earn money, a wife's job is to look after the home and family," the TV distance would fall from 0.148 to 0.063, a reduction of 0.085, representing a 57\% decrease in cultural distance. This counterfactual exercise reveals that migration timing explains over half of the observed divergence for this dimension of gender attitudes. Similarly, for the statement about family suffering when mothers work, the TV distance would decline by 0.054 (a 39\% reduction). These findings suggest that early childhood exposure to UK norms plays a substantial role in shaping immigrants' gender attitudes.

Third, our MTVD estimates indicate that a uniform one-year increase in migration age across the entire immigrant population would increase cultural distance by 0.014-0.034 units, across questions 1, 2, and 4. The largest marginal effect (MTVD = 0.034) occurs for attitudes about traditional household division of labor. To interpret the magnitude: this implies that delaying migration by one year increases the minimum fraction of immigrants who would need to change their responses by approximately 3.4 percentage points. Put differently, policies or circumstances that induce families to migrate one year earlier would close roughly one-fifth to one-quarter of the baseline cultural gap on this dimension. 

Finally, we document heterogeneity by cultural origin. For immigrants from Western Europe, the United States, or Canada, countries with gender norms relatively similar to the UK, we find essentially no effect of migration age on convergence. The Total Variation distance between these immigrants and UK-born respondents is already small (0.08-0.14), and the counterfactual reduction from setting migration age to zero is statistically indistinguishable from zero. By contrast, for immigrants from the rest of the world, the baseline TV distance is substantially larger (0.19-0.20 for three of the four questions), and the counterfactual reduction from earlier migration is both statistically significant and economically large. For example, among non-Western immigrants, setting migration age to zero would reduce the TV distance for traditional household roles from 0.20 to 0.10—a 50\% decline. The MTVD for this group is 0.039, compared to effectively zero for Western immigrants. This heterogeneity is consistent with a cultural exposure mechanism: when origin-country norms diverge more from the destination, the timing of exposure matters more.

\paragraph{Related Literature}  Our paper connects three literatures: (i) work on how \emph{age at migration} shapes integration outcomes, (ii) research on \emph{cultural assimilation and exposure} over the life course, and (iii) studies of \emph{immigrants and gender norms}.

A large economics literature shows that arriving earlier in the destination country improves standard integration outcomes, consistent with childhood being a sensitive period for the accumulation of destination-specific human capital. Classic evidence emphasizes language: \citet{BleakleyChin2004} use variation in age at arrival interacted with origin-language distance to identify large effects of childhood arrival on English proficiency and downstream socioeconomic outcomes. Related work documents sharp gradients in education and earnings by arrival age, including in Canada and Europe \citep{SchaafsmaSweetman2001,VanOursVeenman2006,Gonzalez2003}. More recent research exploits richer administrative data and within-family designs to strengthen causal interpretation. For example, \citet{Bohlmark2008} study siblings in Sweden and show that arriving later substantially harms school performance, while \citet{Hermansen2017} documents large long-run penalties for later childhood arrival. Closest to our identification strategy, \citet{LemmermannRiphahn2018} estimate the causal effect of age at migration on youth educational attainment using a sibling fixed-effects design.

Beyond economic outcomes, a broad literature studies assimilation as the evolution of language, identity, and social ties. \citet{BleakleyChin2010} show that age at arrival affects not only English proficiency but also measures of social assimilation, such as marriage decisions, fertility and residential location, consistent with exposure during formative years shaping integration trajectories. Sociological work similarly emphasizes that migration during different life stages generates distinct ``generational cohorts'' and patterns of acculturation \citep{Rumbaut2004}. Complementary evidence on assimilation over longer horizons appears in work on historical immigration, which documents substantial but incomplete convergence in economic outcomes across generations and cohorts \citep{AbramitzkyBoustanEriksson2014}, and in research emphasizing the imperfect portability of skills across destinations \citep{Friedberg2000}. Taken together, these studies motivate treating age at arrival as a proxy for the timing and intensity of exposure to the host-country environment. 

A separate literature studies how beliefs about gender roles persist and evolve across contexts. Historical perspectives argue that gender norms are deeply rooted in economic and cultural conditions \citep{AlesinaGiulianoNunn2013,InglehartNorris2003}. \citet{Fernandez2007} shows that source-country culture predicts outcomes and beliefs among descendants of immigrants and \citet{FernandezFogli2009} document that gender-related beliefs and behaviors (e.g., female work and fertility) partly reflect inherited cultural attributes from the origin country. More broadly, models and evidence on cultural transmission emphasize that preferences can persist across generations through socialization within families and communities \citep{BisinVerdier2001,FarreVella2013}. Direct evidence on immigrant gender attitudes points to an interaction between inherited culture and host-country assimilation forces: \citet{BlauEtAl2015} study gender roles of immigrants, by tracking fertility decisions and education choices, among immigrants and find patterns consistent with both selection and assimilation, while \citet{RoderMuhlau2014} show that immigrants in Europe exhibit gender attitudes that shift toward host-country norms with exposure. 


\paragraph{Roadmap.} The remainder of the paper proceeds as follows. Section \ref{sec:data} describes the data. Section \ref{sec:empircalStragey} outlines our measurement framework and empirical strategy. Section \ref{sec:results} presents results and robustness checks. Section \ref{sec:conclusions} concludes.

\section{Data}\label{sec:data}
\subsection{Data and Sample Construction}
We use the UK Household Longitudinal Study (Understanding Society; henceforth, UKHLS), a nationally representative panel survey of UK households fielded from 2009 to 2021. The survey is particularly valuable for our purposes because it added an Immigrant and Ethnic Minority Boost sample in 2015, improving coverage of the immigrant population.  We identify immigrants using the survey’s information on whether respondents report a year of arrival to the UK and their country of birth, and we measure age at arrival by subtracting year of birth from the reported year of arrival. We focus on childhood arrivals by restricting attention to individuals whose implied arrival age is no greater than eighteen, aligning the analysis with the idea that exposure during formative years is the relevant margin for cultural socialization. 

As we discuss in section \ref{sec:empircalStragey}, our empirical design compares siblings within the same family, absorbing all time-invariant family background that jointly shapes migration decisions and gender attitudes. We therefore identify sibling links using the UKHLS family relationship matrix, that records kinship ties between panelists, and we build a family identifier that groups together all individuals connected through those reported sibling relationships.  Therefore, the estimation sample consists of families with at least two siblings observed in the attitude modules. This yields roughly three hundred immigrant families and about seven hundred sibling responses per question.  We follow a similar procedure when identifying families of UK-born local panelists, which yields around 2,775 families and 6,250 siblings (See table \ref{desc_locals} in the Appendix).

To measure gender attitudes, we use four questions that were administered in waves 2, 4, and 10 of the panel.  Respondents rate their agreement on a five-point scale (from strongly disagree to strongly agree) with the following four statements:
\begin{itemize}
    \item whether a pre-school child is likely to suffer if their mother works
    \item whether family life suffers when the woman has a full-time job
    \item whether both spouses should contribute to household income
    \item whether a husband’s job is to earn money while a wife’s job is to look after the home and family
\end{itemize}
When respondents answer the same statement in multiple waves, we use the most recent observed response so that each individual contributes a single, up-to-date attitude measure in the main analysis. 

Finally, we study heterogeneity by origin-country context by combining UKHLS with external information on gender norms in sending countries.  For each origin country, we use Wave 7 of the World Values Survey (conducted 2017–2020) to construct an origin-country profile based on average responses to multiple gender-related statements spanning beliefs about mothers’ paid work and child wellbeing, gendered claims about political and business leadership, education priorities for boys versus girls, norms under job scarcity, and beliefs about marital strain when wives earn more.  We then summarize cultural proximity to the UK by computing the distance between each origin-country response vector and the corresponding UK vector, and we use this distance to split the sample into culturally “closer” and “farther” origin-country groups for heterogeneity analysis (See Appendix for more details \ref{appendix:similarity}). 

In addition to our main results, we also demonstrate that the questions about attitudes are also correlated with within-household specialization in both childcare and routine housework. We construct a respondent-level dataset, from the UKHLS data, that links the gender-attitudes items to measures of who mostly performs core housework and core childcare tasks. We restrict attention to individuals who answered the gender-attitudes questions, and limit the sample to respondents who are observed living with a spouse or partner and who have at least one child under age 18 in the household. Because the domestic-division questions are not asked (or not answered) in every wave for every respondent, we collapse the panel information to a single observation per person by using the most recent non-missing response available for each domestic task and for the household-composition restrictions. This produces a consistent “partnered with children” sample for which both housework and childcare allocations are measured and comparable across the attitude questions.

The main outcomes are binary indicators for whether the woman in the couple is reported to do most of a given task (cooking, cleaning, laundry/ironing, or specific childcare activities). These indicators are constructed from respondents’ reports about who typically performs each task: the outcome equals one when the responses indicate that the woman “usually” or “always” performs the task, equals zero when the man “usually” or “always” performs it or when the task is reported as shared about equally, and is set to missing when the task is performed by someone else, by someone outside the household, is not applicable, or is otherwise missing.

\subsection{Empirical Patterns}

Table \ref{table:descriptves} provides descriptive statistics for the siblings sample used to address the four survey questions. The first row indicates that the average migration age is 5 years, with the average migration year being 1998. Respondents were, on average, 24 years old at the time of the interview, with most responses collected during the 10th wave of the survey. The average family includes 2.8 children, counting all siblings regardless of survey participation. The sample comprises 305 families with at least two children who responded to the gender role questions, yielding approximately 729 sibling observations. Importantly, the sample is consistent across all questions.

Table \ref{table:descriptves} further shows that the sex ratio of respondents is roughly balanced. Table \ref{table:originCont} in the Appendix details the geographic origins of the immigrants, revealing that the majority originate from Asia, followed by Europe and Africa. Finally, the geographic distribution of the siblings sample closely mirrors that of the broader UKHLS immigrant population under the age of 18. 

\begin{table}[!h]
\footnotesize
    \centering
        \resizebox{\textwidth}{!}{
        \footnotesize
    \begin{tabular}{lcccc}
\hline
\textbf{} & \makecell{\textbf{Pre-School} \\ \textbf{Suffers} \\ \textbf{if Mother Works}} & \makecell{\textbf{Family Suffers} \\ \textbf{if Mother Works} \\ \textbf{Full Time}} & \makecell{\textbf{Husband and Wife} \\ \textbf{Should Contribute} \\ \textbf{to HH Income}} & \makecell{\textbf{Husband Should} \\ \textbf{Earn,} \\ \textbf{Wife Stay at Home}} \\
\hline
Migration Age  & 5.0 & 5.0 & 4.9 & 5.0 \\ 
  & (5.3)  & (5.3)  & (5.3)  & (5.3)  \\ 
Migration Year  & 1998.2 & 1998.2 & 1998.2 & 1998.2 \\ 
  & (9.0)  & (9.0)  & (9.0)  & (9.0)  \\ 
Age at Interview   & 23.9 & 24.0 & 23.9 & 23.9 \\ 
  & (7.6) & (7.6) & (7.7) & (7.6) \\ 
Sex  & 0.51 & 0.51 & 0.51 & 0.51 \\ 
  & (0.50) & (0.50) & (0.50) & (0.50) \\ 
Wave b   & 0.12 & 0.12 & 0.12 & 0.12 \\ 
  & (0.33) & (0.33) & (0.33) & (0.33) \\ 
Wave d   & 0.30 & 0.30 & 0.30 & 0.30 \\ 
   & (0.46) & (0.46) & (0.46) & (0.46) \\ 
Wave j   & 0.57 & 0.57 & 0.58 & 0.57 \\ 
   & (0.49) & (0.49) & (0.49) & (0.49) \\ 
Average Score  & 2.99 & 3.27 & 2.07 & 3.56 \\ 
  & (1.06) & (1.14) & (0.92) & (1.19) \\ 
Family Size  & 2.8   & 2.8   & 2.8   & 2.8   \\ 
      & (1.1)   & (1.1)   & (1.1)   & (1.1)   \\ 
Num. Families  & 305.0   & 304.0   & 303.0   & 304.0   \\ 
Num. Obs  & 729.0  & 727.0  & 726.0  & 728.0  \\ 
\bottomrule
\end{tabular}

    }
    \caption{Descriptive Statistics}\label{table:descriptves}
    \label{tab:my_label}
    \begin{minipage}{\linewidth}
    \footnotesize
    \justifying
    \textit{Note:} Descriptive statistics for the immigrant-siblings sample, shown separately for each of the four attitude questions (columns (1)–(4): Pre-School Suffers if Mother Works; Family Suffers if Mother Works Full Time; Husband and Wife Should Contribute to Household Income; Husband Should Earn, Wife Stay at Home). Rows report the mean of key variables (Migration Age, Migration Year, Age at Interview, Sex, survey-wave indicators, Average Score, and Family Size), with standard deviations in parentheses. The bottom rows report the number of families and number of observations used for each column.
    \end{minipage}
\end{table}


Figure \ref{fig:answerDistribution} illustrates the distribution of responses in our immigrant siblings sample (the blue bars), compared to those of locals (the green bars). Substantial differences emerge across the four questions, with immigrants showing a higher likelihood of agreeing with the first, second, and fourth statements relative to locals, suggesting that immigrants are less likely to hold egalitarian views toward gender roles.

Table \ref{table:responseByGroup} in the Appendix provides a detailed comparison of response distributions across four groups: the immigrant siblings sample, the overall immigrant population in the UKHLS data, locals, and a subset of locals restricted to those with at least one sibling in the UKHLS sample. The results indicate that immigrant siblings are more likely to disagree with the statements compared to the overall immigrant population. Moreover, when comparing immigrant siblings to the equivalent local sibling sample, immigrant siblings show a greater tendency to agree with the four statements. A similar pattern is observed when comparing the overall immigrant population in the sample to the overall local population.  


An important difference between our immigrant siblings sample and the overall immigrant sample is respondent age. As shown in Appendix Table \ref{table:ageByGroup}, the median age in the immigrant siblings sample is 22, compared with 43 in the broader immigrant population. A similar pattern holds for locals: the median age in the local-siblings sample is 23, compared with 48 in the general local population. Therefore, restricting the sample to respondents with siblings yields a substantially younger local (and immigrant) sample.

Age is strongly associated with responses to the four survey statements. Appendix Figure
\ref{fig:immigrants_Age_answers_relation} plots all immigrant respondents’ age at the time of the interview against the
\emph{mean coded response} for each statement, where responses are coded as: ``Strongly agree'' = 1,
``Agree'' = 2, ``Neither agree nor disagree'' = 3, ``Disagree'' = 4, and ``Strongly disagree'' = 5 (so lower
average scores indicate greater agreement). Among the general immigrant population, older respondents
tend to have lower average scores (greater agreement) on the first, second, and fourth statements, for
which agreement corresponds to more traditional (less gender-equal) views. By contrast, for the third
statement (``Husband and wife should contribute to household income''), agreement is the more
egalitarian position, and older respondents tend to have higher average scores (less agreement). Figure \ref{fig:locals_Age_answers_relation} displays similar patterns among local respondents. Taken
together, these patterns imply that, as people age, they are more likely to hold and express beliefs that
are not gender-equal, as captured by these items and their coding. Equivalent patterns are evident among
local respondents. Therefore, since our identification strategy restricts the sample to respondents with
siblings, we maintain comparability by applying the same restriction to locals and immigrants, and we
compare immigrant siblings to local siblings in the main analysis.

Figure \ref{fig:migrationAge_allMig} presents a binscatter showing the relationship between migration age
and the mean coded response for all immigrants in the UKHLS sample who moved to the UK before the
age of 18, using the same coding as above (lower average scores indicate greater agreement). The figure reveals a strong negative relationship for the first, second, and fourth statements, indicating that immigrants who migrate at older ages are more likely to agree with these statements. This suggests a greater likelihood of holding traditional beliefs about gender roles, such as the view that women should not participate in the
labor market and should focus on domestic responsibilities. The relationship is particularly strong for questions addressing whether the family and children suffer when women work and whether the husband should work while the wife takes care of the home.  The interpretation differs for the statement ``Husband and wife should contribute to household income,'' for which agreement is the more egalitarian position. For this item, the association with migration age is
negative but notably weaker, implying at most a small increase in agreement with this egalitarian statement (and thus a much weaker relationship than for the items where agreement reflects traditional views).

Figure \ref{fig:migrationAge_siblings} shows the same binscatter for our siblings sample. While the pattern reflects the negative association between agreement with the statements and migration age observed in the broader immigrant sample, the relationship is weaker. Specifically, the association between migration age and agreement with the statement that both husband and wife should contribute to household income is almost negligible in the siblings sample.

These descriptive patterns imply that migrating at an earlier age is associated with holding less traditional gender-role attitudes. However, these descriptive relationships should not be interpreted causally: migration age is correlated with many other determinants of gender attitudes, including cohort and period of arrival, country (and region) of origin, parental education and religiosity, selective migration and return migration, and the local environments immigrants are exposed to upon arrival (e.g., neighborhood composition and school quality). 

\begin{figure}[!h]
    \centering
    \includegraphics[scale=1]{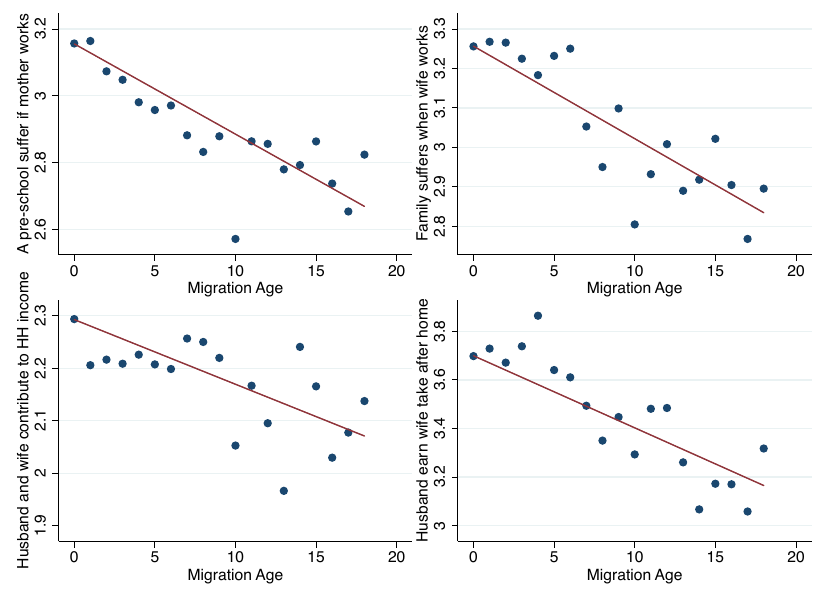}
    \caption{Relation between migration age and attitude towards Gender role}
    \label{fig:migrationAge_allMig}
    \begin{minipage}{\linewidth}
    \footnotesize
    \justifying
    \textit{Note:} Four-panel binscatter of migration age (x-axis, in years) against the mean coded response (y-axis) for each gender-role statement, using all immigrants in UKHLS who moved to the UK before age 18. Responses are coded on a five-point scale: Strongly agree = 1, Agree = 2, Neither agree nor disagree = 3, Disagree = 4, Strongly disagree = 5. Each panel overlays a fitted linear trend line on the plotted age-specific mean responses.  
    \end{minipage}
\end{figure}

Next we examine how immigrants’ gender-role attitudes vary with the gender norms of their country of origin. Appendix Figure \ref{fig:distance_allMig} plots immigrants’ agreement with each survey statement against the distance between origin-country gender attitudes and those in the UK. Immigrants from countries with norms that are further from the UK tend to agree more with the statements, consistent with more traditional gender-role views. The exception is the statement that both spouses should contribute to household income, for which the relationship is only weakly negative. Appendix Figure \ref{fig:distance_siblings} shows the same pattern in our siblings sample. 

Finally, Figure \ref{fig:divsionHH} shows that responses to the gender-attitudes questions are reflected in concrete household specialization. We focus on UKHLS respondents who both answered the gender-attitudes questions and are observed as married with children in at least one wave of our sample, but not necessarily immigrants. For this population, we measure whether the woman (rather than the man) performs most of various tasks related to childcare and core household chores (cleaning, cooking, and laundry). Specifically, we estimate the coefficient vector $\beta_r$ in the regression
\[
\text{Woman does task k}_i = \alpha +
\sum_{\substack{
r \in \{\text{Agree, Neither agree nor disagree,}\\
\text{Disagree, Strongly disagree}\}
}}
\beta_{r,k}\,\mathbbm{1}[\text{Response}=r] + \varepsilon_i .
\]
where $\text{Woman does task k }_i$ is an indicator for whether the woman in respondent $i$'s household performs most of the given task $k$, and $\mathbbm{1}[\text{Response}=r]$ indicates that the respondent selected response category $r$ (with the omitted category being “Strongly agree”). The resulting coefficients trace out a clear gradient: moving away from traditional responses is associated with a lower likelihood that the woman performs most routine chores and childcare. Importantly, the direction of this gradient lines up with the normative content of the items. For the three statements where “Strongly agree” corresponds to a more traditional view of gender roles (Panels \ref{fig:a}, \ref{fig:b}, and \ref{fig:d}), the coefficients become increasingly negative as respondents shift toward disagreement, indicating less household specialization onto women. By contrast, for the statement “Husband and wife should contribute to household income” (Panel \ref{fig:c}), “Strongly agree” is the more egalitarian position, and the pattern correspondingly flips sign: moving toward disagreement is associated with greater specialization onto women. Taken together, these correlations suggest that the attitude measures are not merely expressive survey responses; they covary in a systematic and substantively intuitive way with within-household specialization in the domain most directly connected to the questions.

\begin{figure}[!htbp]
\centering
\begin{subfigure}[t]{0.48\textwidth}
  \centering
  \includegraphics[width=\linewidth]{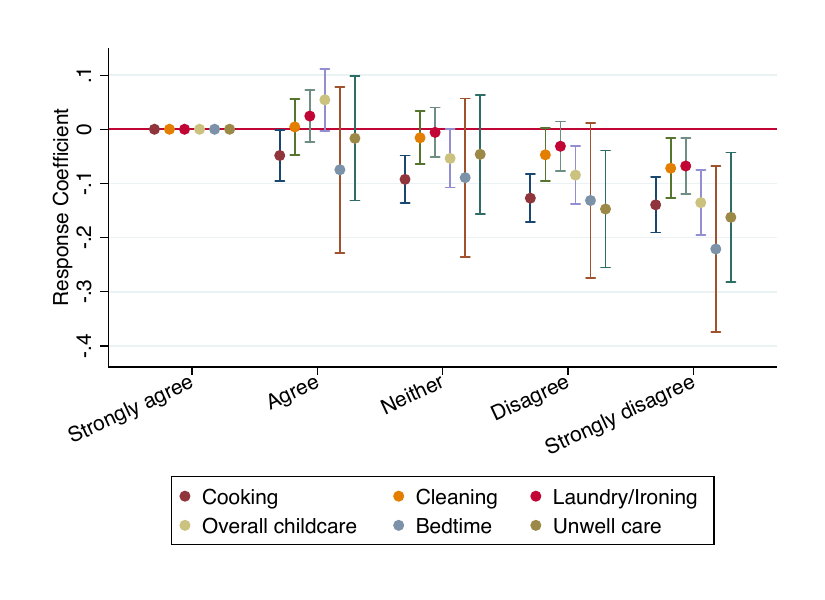}
  \caption{Pre-school child suffers if mother works}
  \label{fig:a}
\end{subfigure}\hfill
\begin{subfigure}[t]{0.48\textwidth}
  \centering
  \includegraphics[width=\linewidth]{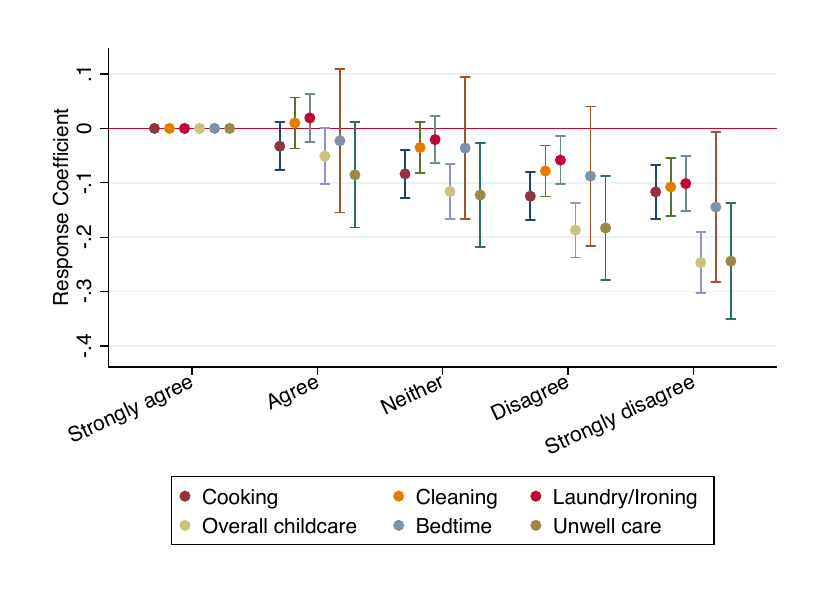}
  \caption{Family suffers if mother works}
  \label{fig:b}
\end{subfigure}

\vspace{0.6em}

\begin{subfigure}[t]{0.48\textwidth}
  \centering
  \includegraphics[width=\linewidth]{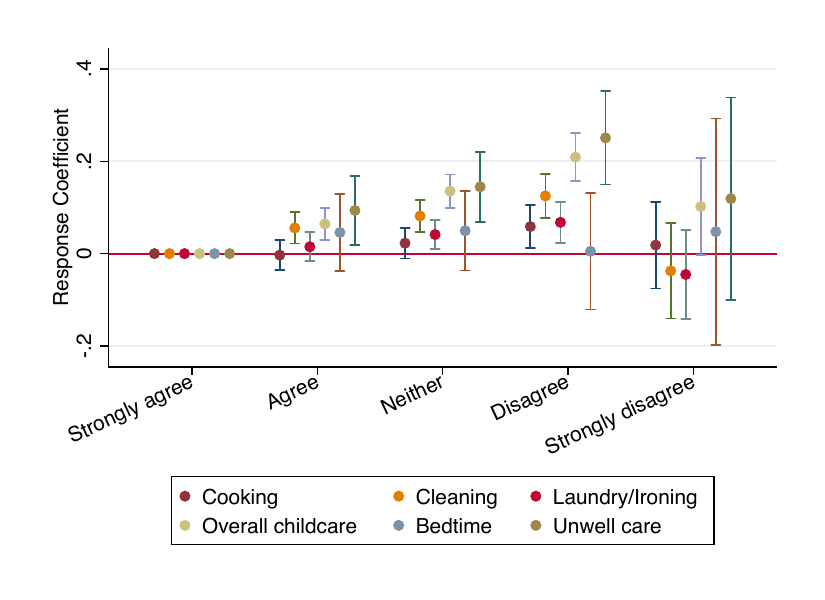}
  \caption{Husband and wife should contribute to household income}
  \label{fig:c}
\end{subfigure}\hfill
\begin{subfigure}[t]{0.48\textwidth}
  \centering
  \includegraphics[width=\linewidth]{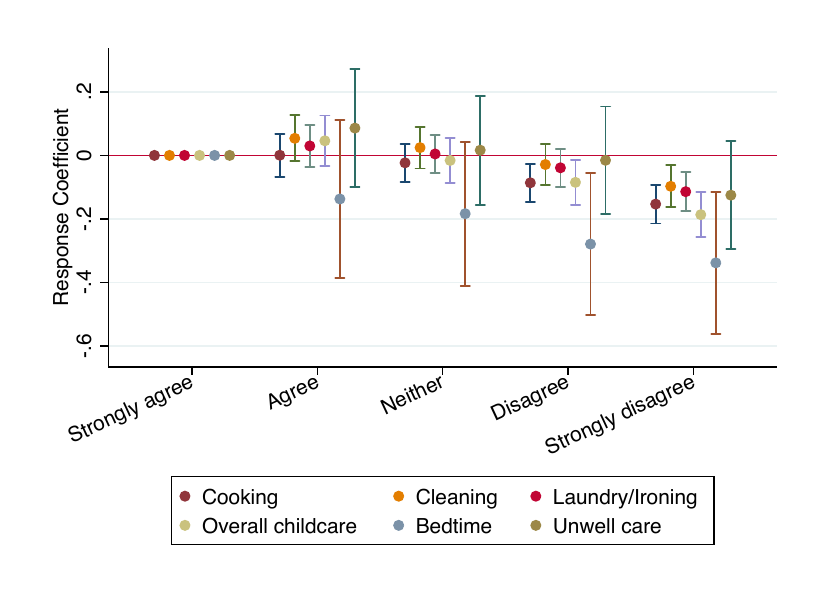}
  \caption{Husband should earn; wife should stay at home}
  \label{fig:d}
\end{subfigure}

\caption{Relationship Between Gender-Role Attitudes and Household Specialization}
    \begin{minipage}{\linewidth}
    \footnotesize
    \justifying
    \textit{Note:} Four-panel coefficient plot linking gender-role attitudes to within-household task specialization in the “partnered with children” sample (respondents who answered the attitude questions and are observed married with children in at least one wave; not restricted to immigrants). For each panel (one attitude statement), the x-axis lists response categories, and the y-axis (“Response Coefficient”) reports estimated coefficients $\beta_{r,k}$  from a regression of an indicator for whether the woman does most of task \(k\) on dummies for response categories \(r \in \{\text{Agree, Neither, Disagree, Strongly disagree}\}\), with “Strongly agree” omitted. Separate series (colored markers) correspond to tasks: cooking, cleaning, laundry/ironing, overall childcare, bedtime, and unwell care. Vertical whiskers around each point the robust standard errors, with 95\% confidence. 
\end{minipage}
\label{fig:divsionHH}

\end{figure}

\section{Identification and Empirical Strategy}\label{sec:empircalStragey}
Our goal is to estimate the causal effect of age at arrival on individuals’ attitudes toward gender roles, but a simple comparison of early versus late arrivals is unlikely to recover it because migration timing is confounded with family characteristics and circumstances. On one hand, highly educated parents may move to the UK to maximize their children’s educational opportunities and therefore migrate when children are young; since higher education is correlated with progressive values, this can make earlier arrival appear associated with more egalitarian attitudes even absent a causal effect. On the other hand, financially constrained or crisis-driven families may migrate only after delays, so their children arrive older; if those constraints or disruptions also shift norms (e.g., through economic necessity or weakened traditional enforcement), later arrival may appear more progressive for reasons unrelated to age at arrival. Either way, unobserved determinants of migration timing correlate with attitudes, so the naive early–late comparison can be biased in either direction.

To address these concerns, we adopt a sibling design (e.g., \cite{chetty2018impacts, LemmermannRiphahn2018}), comparing siblings within the same family who differ in their age at arrival in the UK. The appeal of this approach is that it holds fixed the entire bundle of family background and parental selection into migration: siblings share parents, an origin country, and the same migration episode, so the comparison is not between “different kinds of families,” but between children who experienced that same move at different points in childhood. In this sense, within-family variation lets us control for family-level confounding factors such as stable family characteristics, preferences, and long-run plans—and focus on differences in exposure that arise within the family.

To implement this within-family comparison, we employ a family fixed-effects approach. Let \(Y_i\) denote a survey response outcome for individual \(i\), \(MigAge_i\) the age at which individual \(i\) arrived in the UK, and \(f(i)\) the family of individual \(i\). We estimate the parameter \(\beta\) in the following linear model: 
\[
Y_i = MigAge_i\beta + \gamma_{f(i)} + X_i\delta + u_i,
\]
where \(\gamma_{f(i)}\) is the family fixed effect for \(i\), \(X_i\) represents additional controls discussed in the robustness section, and \(u_i\) captures unobserved factors influencing the outcome. This specification compares siblings within the same family while absorbing any time-invariant family-level determinants of both migration timing and gender attitudes.

Our main identification assumption for $\beta$ is that, conditional on family membership (and on the additional controls $X_i$ included in the regression), the within-family component of age at arrival is uncorrelated with unobserved individual determinants of gender attitudes:
\begin{equation}\label{eq:assumption1}
E\!\left[\Big(MigAge_i-E[MigAge_i\mid f(i),X_i]\Big)\,u_i \,\middle|\, f(i),X_i\right]=0.
\end{equation}
The intuition is straightforward in the common case where a family migrates once: siblings' arrival ages differ largely because they were born in different years relative to the family's migration date. The identifying claim is therefore \emph{not} that migration timing is random in the population, but that, once we condition on the family (and on $X_i$), it is not systematically the case that the sibling who happened to arrive relatively older (or younger) is also the sibling with unobserved traits that would have pushed their gender attitudes in a particular direction regardless of migration age.\footnote{Condition \ref{eq:assumption1} is the orthogonality restriction that identifies $\beta$ in the family fixed-effects regression with controls. Our identifying story can be better expressed in the ``design-based'' form, $E[MigAge_i \mid u_i,f(i)]=\mu_{f(i)}$, which formalizes the idea that migration timing is determined at the family level and is not driven by child-specific latent traits. Put differently, conditional on family information, learning additional child traits does not improve prediction of when the family migrated and therefore at what age each child arrived (see, e.g., \citealp{borusyak2024negative}). We maintain Condition \ref{eq:assumption1}, as it is the weaker assumption, implied by the design-based assumption. }

This identifying story is plausible because the forces that determine when a household migrates are typically household-level and therefore common to all siblings. Migration timing is largely governed by considerations that operate at the household level and apply similarly to all children in the family. Families typically migrate in response to a combination of parental opportunities and constraints—labor-market prospects, visa availability, the cost of relocating, housing constraints, the presence of relatives or social networks in the destination, and the broader political or economic environment in the origin country—together with preferences over where to raise children and how to invest in them. These determinants can be strongly related to both the decision to migrate and to parental values, but crucially they are shared across siblings and therefore absorbed by the family fixed effect. Moreover, even when parents have child-centered motives for migrating (for example, better schooling or safer neighborhoods), these motives commonly target the well-being of the household as a whole rather than the idiosyncratic traits of one particular child. In that setting, within-family differences in age at arrival primarily reflect birth timing relative to a family-level migration date, not strategic re-timing of migration in response to sibling-specific latent traits.

Moreover, the migration date is often influenced by discrete events that are plausibly orthogonal to child-specific unobservables, such as a job offer that materializes at a particular time, the resolution of administrative and visa processes, or a sudden deterioration (or improvement) in conditions in the origin country. Even absent such shocks, migration often involves planning horizons and constraints that make fine-tuning the move to a specific child's latent disposition unlikely: parents can rarely condition the timing of a move on subtle within-family differences in children's attitudes, and those attitudes are typically not well revealed in early childhood, when much of the identifying variation in age at arrival is generated. For our setting, the identifying restriction therefore amounts to a ``no within-family sorting on latent determinants of gender attitudes'' assumption: conditional on family membership and the basic controls $X_i$ (which capture the most direct mechanical within-family confounds), the residual variation in age at arrival across siblings is as good as random with respect to $u_i$.

Several potential threats to this identifying restriction warrant attention. A first concern is that within-family differences in \(MigAge_i\) are naturally intertwined with birth order: older siblings tend to be older at arrival, and birth order itself may shape attitudes through parental expectations, household responsibilities, or shifting family circumstances. For example, families may face tighter budgets early on and accumulate wealth or social capital over time, so older siblings could grow up under greater strain while younger siblings benefit from improved conditions; these forces could generate within-family attitude gaps even if arrival age had no causal effect. To address this, our robustness checks include an indicator for being the oldest sibling, which helps separate birth-order effects from the role of exposure age. Closely related, siblings are often interviewed at different ages, and age at interview may independently affect reported attitudes through life-cycle changes in beliefs; moreover, siblings may be surveyed in different waves, and broader period effects could also shift responses. If age-at-interview or interview timing is mechanically correlated with within-family differences in \(MigAge_i\), estimates of \(\beta\) could confound exposure age with these measurement and period channels. We therefore examine robustness to flexible controls for age at interview and to accounting for interview timing in the survey.

A second concern is that child-specific traits could influence parental migration decisions in a way that creates within-family correlation between \(MigAge_i\) and \(u_i\). For instance, parents might be more likely to migrate when the eldest child exhibits high academic ability or ambition—traits that may also correlate with more progressive values—while similar traits in younger siblings are less salient for the timing of the move. In that case, migration age would be systematically related to unobserved determinants of attitudes, violating (\ref{eq:assumption1}). This concern is attenuated if migration decisions respond to family-level opportunities or constraints, or to the well-being of the household as a whole, rather than to idiosyncratic traits of a particular child; our framework relies on this collective-decision view of migration timing.

Finally, older children could themselves influence the decision to migrate. An older sibling might advocate for migration to improve their own opportunities, or parents might weigh the older child’s preferences more heavily, again linking within-family \(MigAge_i\) to unobserved determinants of attitudes. If present, this channel would violate (\ref{eq:assumption1}). Our empirical strategy therefore rests on the premise that migration decisions are primarily parent-driven and that, within families, the residual determinants of gender attitudes do not systematically sort with which sibling happened to arrive at a relatively older age.

\subsection{Measuring Convergence in Attitudes Toward Gender Roles}\label{sec:convergenceMeasure}  
Our objective is to examine whether the age at which individuals migrate influences the alignment of their attitudes and beliefs about gender roles with those of the corresponding local population in the UK. Therefore, rather than emphasizing a single “average effect” under an arbitrary coding of the response scale, we focus on whether changes in the migration-age distribution among immigrants causally affect the overall similarity between immigrant and UK-born response distributions. \citep{chernozhukov2013inference}.

In what follows, we first motivate the use of Total Variation Distance as a measure of similarity between groups in our analysis. Next, we discuss how to measure the effect of changes in the distribution of migration age on the similarity between immigrants and locals, and finally, we present the empirical model that will guide our estimation.

\subsubsection{Total Variation Distance and The Worst-Case Average Difference  }\label{sec:totalVariationWCAD}

Researchers often seek to compare survey responses across groups. A widely used method for assessing similarity in responses to multiple-choice questions involves assigning numerical values to each possible answer and comparing the group averages. For example, in our setup, respondents choose from "Strongly Disagree," "Disagree," "Neither Agree nor Disagree," "Agree," or "Strongly Agree," which could be assigned values from 1 to 5. The difference in average scores between groups is then used as a measure of dissimilarity, with smaller differences indicating greater similarity.  

Although intuitive, this approach has a significant limitation in that assigning numerical values to responses is inherently arbitrary. Different choices of numerical values impose varying weights on the differences in response shares across groups. For instance, assigning the value 5 to "Strongly Agree" and 1 to "Strongly Disagree" means that differences in the share of respondents choosing "Strongly Agree" are weighted five times more heavily than differences in the share of respondents choosing "Strongly Disagree."  

To formalize this, let \( a^q_i \in \mathcal{A}^q \) represent the answer that individual \( i \) provides to question \( q \), where \( \mathcal{A}^q \) is the set of possible responses to question \( q \). Let \( h: \mathcal{A}^q \to \mathbb{R} \) be a function assigning a numerical value to each response. The expected numerical response for a group is then:  
\[
\mathbb{E}[h(a_i^q) \mid \text{Group}] = \sum_{a \in \mathcal{A}^q} h(a) \Pr(a^q_i = a \mid \text{Group}).
\]
The difference in expected responses between immigrants and locals is:  
\[
\mathbb{E}[h(a^q_i) \mid Immigrants] - \mathbb{E}[h(a^q_i) \mid Locals] 
= \sum_{a \in \mathcal{A}^q} h(a) \left[ \Pr(a^q_i = a \mid Immigrants) - \Pr(a^q_i = a \mid Locals) \right].
\]
This expression demonstrates that the difference in expected numerical responses is a weighted sum of the differences in response probabilities between the two groups, where \( h(a) \) explicitly determines the weights. Consequently, the choice of \( h(a) \) directly influences the measured difference across groups.  

Choosing \( h(a) \) depends on the specific objectives of a decision-maker, and different policymakers may prioritize varying aspects of the response distribution based on their policy goals. For instance, one policymaker might be particularly concerned about differences in the share of individuals who "Strongly Agree" with traditional gender roles, viewing strong beliefs as more resistant to change and more influential on behavior. Another policymaker might focus on the proportion of individuals who "Agree," interpreting moderate agreement as indicative of general acceptance and thus more amenable to policy interventions.

Given these diverse priorities and uncertainties about which differences are most critical, relying on a single, arbitrary scoring function \( h(a) \) may fail to capture the differences that matter most to different decision-makers. Our approach draws inspiration from the literature on robust decision-making and worst-case analysis \citep{hansen_sargent_2008, gilboa_schmeidler_1989, brooks2025simplicity}. In robust decision-making, policymakers consider the worst-case scenario to ensure that decisions are effective under the most adverse conditions. Similarly, we consider the \emph{worst-case average difference} between immigrants and locals across all possible scoring functions \( h(a) \) that map responses to values within a bounded interval. This approach ensures robustness by capturing the largest possible average difference across all bounded scoring functions, providing a  measure of group dissimilarity robust to different weighting. 

Formally, let \( p^q_{Immigrants} \) and \( p^q_{Locals} \) denote the distributions of answers among immigrants and locals, respectively. We then define the \emph{worst-case average difference (WCAD)} between immigrants and locals as the maximum discrepancy in average responses across all possible scoring functions \( h(a) \) that map responses to values within a bounded interval, such as \([0, 1]\):  
\begin{equation}\label{eq:worst_case_difference} 
\sup_{h: \mathcal{A}^q \to [0,1]} \left| \mathbb{E}[h(a^q_i) \mid Immigrants] - \mathbb{E}[h(a^q_i) \mid Locals] \right|.  
\end{equation}

This measure captures the maximum possible discrepancy in average responses between immigrants and locals across all valid scoring functions \( h(a) \), bounded between 0 and 1.\footnote{These bounds can easily be scaled} By focusing on the worst-case scenario, we ensure that the measure accounts for the greatest potential divergence in attitudes, regardless of the specific weighting scheme employed. If WCAD decreases, it implies that the worst-case discrepancy from any perspective has also decreased.

Although the response options are presented on an ordered scale, our distance metric deliberately does \emph{not} exploit this ordinal structure. In \eqref{eq:worst_case_difference} we allow $h$ to range over
all bounded recodings of the response categories, without imposing that $h$ respects the natural ordering
(or any particular spacing between adjacent categories). This is natural in our application: for attitudes
toward gender roles there is no clear consensus among policymakers about which response categories should be
weighted more heavily, so we remain fully agnostic about how decision-makers value each category. A direct
implication is that WCAD is invariant to relabelings of the response categories and therefore treats the categories as labels rather than as points on an ordered scale.\footnote{If one wishes to incorporate
the ordinal nature of Likert responses, one can restrict $h$ to be bounded and monotone with respect to the
response ordering. In that case, the supremum in \eqref{eq:worst_case_difference} is equivalent to the
Kolmogorov distance between the immigrants' and locals' response distributions (i.e., the
maximum difference in cumulative response shares across response values). See Appendix \ref{app:Ordinal} for additional discussion}

An appealing property of WCAD is its equivalence to the Total Variation distance between the response distributions of migrants and locals (e.g., \cite{levin2017markov}, Proposition 4.5).\footnote{The standard claim is that the total variation distance is equivalent to half the maximum difference in expectations over the set of bounded functions between \(-1\) and \(1\). However, we can easily extend this result to functions bounded between \(0\) and \(1\). Specifically, for any function \( g \) bounded between \(-1\) and \(1\), we can define a corresponding function \( f(x) = \frac{g(x)+1}{2} \), which is bounded between \(0\) and \(1\). It follows that:  
\[
\sup_{|g|\leq 1} \left| \mathbb{E}_p[g] - \mathbb{E}_q[g] \right| = 2 \cdot \sup_{0 \leq f \leq 1} \left| \mathbb{E}_p[f] - \mathbb{E}_q[f] \right|.
\]} The TV is straightforward to calculate and is defined as:
\[
TV(p^q_{Immigrants}, p^q_{Locals}) = \frac{1}{2} \sum_{a \in \mathcal{A}} \left| p^q_{Immigrants}(a) - p^q_{Locals}(a) \right|.
\]
In the discrete-support case, TV equals half the sum of absolute differences in the probabilities assigned to each response option. Consequently, to compute WCAD we do not need to optimize over all bounded weighting functions: we can compute it directly from the estimated response distributions via the expression above. The equivalence is thus useful both conceptually (WCAD summarizes the maximal discrepancy in expectations across bounded functions) and practically (it reduces the calculation to a simple function of the observed response shares).

This representation shows that WCAD is numerically identical to a familiar distance between distributions, which makes it easy to compute from the empirical response frequencies. It also allows us to interpret WCAD via a standard probabilistic characterization of TV: the coupling interpretation. Specifically, the TV distance quantifies the smallest proportion of immigrants whose responses would need to change to make their distribution identical to that of the local population. Formally, let \( a^q_{Immigrants} \) and \( a^q_{Locals} \) denote the random variables representing the responses of immigrants and locals to question \( q \). Then the TV distance can also be expressed using the coupling interpretation as:
\[
\text{TV}(p^q_{Immigrants}, p^q_{Locals}) = \inf_{\pi \in \Pi(p_{Immigrants}, p_{Locals})} \Pr_\pi(a^q_{Immigrants} \neq a^q_{Locals}),
\]
where \( \Pi(p_{Immigrants}, p_{Locals}) \) is the set of all joint distributions of \( p_{Immigrants} \) and \( p_{Locals} \). In this context, \( \Pr_\pi(a^q_{Immigrants} \neq a^q_{Locals}) \) represents the minimum probability that an immigrant’s response \( a^q_{Immigrants} \) differs from a local’s response \( a^q_{Locals} \) under the coupling \( \pi \). This infimum directly corresponds to the minimal fraction of immigrants that must alter their responses to eliminate any divergence between the two groups.

In sum, total variation distance has two policy-relevant interpretations: the worst-case difference in average responses, agnostic to any particular value function of the decision maker, and the minimal fraction of responses that must change to achieve alignment. We now define causal parameters describing how shifts in the migration-age distribution affect this distance.


\subsubsection{Measuring the Causal Effect of Migration Age on the Total Variation Distance Between Immigrants and Locals}  
In our empirical analysis, we aim to examine how changing the migration age for the entire population of immigrants would affect the TV distance between immigrants and the local UK population. To measure this, we propose two measures. The first measure evaluates how setting the migration age to zero for all immigrants in our sample would impact the TV distance between locals and immigrants. The second measure considers how a marginal shift in the migration age of the immigrant population would influence the TV and WCAD.  

We begin by analyzing the effect of assuming all immigrants in our sample were born in the UK. Let \( p^{CF,0}_{Immigrants} \) denote the counterfactual choice distribution in a scenario where all immigrant children were not abroad, but in the UK, i.e., assuming we "fix" (\cite{heckman2015causal}) \( MigAge = 0 \). Denote by $a_{i,MigAge=0}^q$ the potential outcome response of individual $i$ to question $q$, if they were born in the UK, then  
\[
p^{CF,0}_{Immigrants}  = 
\int 
P(a^q_{i,MigAge=0} \mid MigAge = x, f(i), \text{Immigrants}) 
\, dF(f(i), MigAge = x \mid \text{Immigrants}),
\]
Then, our first measure of the effect of migration age on attitude convergence is given by:
\[
\Delta TV^0_q = TV(p^q_{Locals}, p^q_{Immigrants}) - TV(p^q_{Locals}, p^{CF,0}_{Immigrants}).
\]
This measures how the TV distance would change in a counterfactual world in which we shift the migration dates of current immigrants (e.g., setting migration age to zero). It can also be interpreted via WCAD as the induced change in an upper bound on the policymaker’s loss—when the policymaker’s value function is bounded—arising from differences between the two populations’ distributions. A positive value of $\Delta TV_q$ implies that this worst-case average difference decreases when migration age is set to zero, suggesting that if all immigrants in the sample were assumed to be born in the UK, their attitudes would become more similar to those of the UK population. Conversely, a negative value indicates that attitudes would diverge further under this assumption.

This measure is a ``global" measure, as it considers a large change in migration age values. It corresponds to a thought experiment that asks: how would the TV or WCAD would change, if instead of the realized migration ages, all immigrant families had children only after migrating to the UK? While this measure is insightful as a conceptual exercise, its interpretation in a policy context is more nuanced. Specifically, it corresponds to a hypothetical policy that restricts migration to families without children, requiring families to have children only after settling in the UK. However, taking this measure at face value assumes that the distribution of families who choose to immigrate to the UK would remain unchanged under such an extreme policy. This assumption is unlikely to hold, as migration decisions would plausibly adjust in response to this policy. Therefore, alongside this global measure, we also explore a more incremental approach by examining the \textit{marginal effect} of migration age on the TV Distance.

The second measure, the Marginal Total Variation Divergence (MTVD), captures the marginal effect of migration age on the divergence between the answer distributions of immigrants and locals. This measure answers the counterfactual question: what would happen to the TV distance between these distributions if the migration age for all immigrants were slightly increased? The MTVD reflects how a small perturbation in migration age, around the current distribution of answers, affects the similarity between the two groups. It assesses whether this marginal increase would lead to greater or lesser divergence in attitudes.

To define the MTVD formally, consider the perturbed distribution:
\[
p_\epsilon(a^q_i \mid Immigrants) = 
\int 
P(a^q_{i,x+\epsilon} \mid MigAge = x, f(i), Immigrants) 
\, dF(f(i), MigAge = x \mid Immigrants),
\]
where \( p_\epsilon(a^q_i \mid Immigrants) \) represents the counterfactual probability that an immigrant chooses response \( a^q_i \) to question \( q \) if the migration age of the entire population increases by \( \epsilon \). As discussed in Section \ref{sec:empircalStragey}, this represents a causal change for the entire population. To measure the MTVD, we calculate:
\[
MTVD = \frac{\partial TV(p_{\epsilon}, p_{locals})}{\partial \epsilon} \bigg|_{\epsilon=0}.
\]
This derivative quantifies how a small increase in migration age affects the similarity between the answer distributions of immigrants and locals.

The MTVD corresponds to a wide range of policies aimed at reducing the time-to-migrate. For example, governments could expedite decision times or encourage families to migrate when their children are young. Alternatively, policies that increase the likelihood of granting asylum to families with younger children could incentivize earlier migration. By measuring the MTVD, we gain insight into how incremental changes in migration timing might influence cultural convergence and alignment with local attitudes. Crucially, from a policy perspective, we assume that marginal changes in children’s age at migration do not alter the composition of families who choose to migrate.

\subsubsection{Empirical Model}   
To operationalize these two measures, we assume a linear probability model. Specifically, we estimate the parameters of the following linear probability model:  
\begin{equation}\label{eq:linearProbModel}
\pr(a^q_i \mid \text{MigAge}, f(i)) = \text{MigAge}_i \beta_{q,a} + \gamma^{q,a}_{f(i)},
\end{equation}  
where \( \beta_{q,a} \) captures the response of the share of immigrants who answered \( a^q_i \) to question \( q \) to changes in migration age, and \( \gamma^{q,a}_{f(i)} \) represents family fixed effects.  

This linear probability model allows us to extend the empirical strategy discussed in Section \ref{sec:empircalStragey} to control for family fixed effects while identifying the causal effect of migration age on mean outcomes. This enables us to explore how migration age affects the similarity between immigrants and locals\footnote{Identifying counterfactuals in nonlinear models, with fixed effects, is notoriously challenging and in most cases is not identified. See, for example, \cite{HonoreKesina2017} and \cite{Chamberlain2010}.}.  

Using this model, we can estimate our first measure as:  
\[
\Delta TV^0_q  = \underbrace{\frac{1}{2}\sum_{a \in \mathcal{A}^q} \left| P(a_i^q|Immigrants) - P(a_i^q \mid \text{Locals}) \right|}_{\text{Observed TV}} -   \underbrace{\frac{1}{2} \sum_{a \in \mathcal{A}^q} \left| \int_{i \in Immigrants} \gamma_i^{q,a} \, di - P(a_i^q \mid \text{Locals}) \right|}_{\text{CF TV}},
\]
where the integral is taken over the population of immigrants.  Similarly, to compute the MTVD, we combine Equation \ref{eq:linearProbModel} with the definition of the Total Variation Distance to obtain:  
\be
MTVD = \frac{1}{2} \sum_{a^q \in \mathcal{A}^q} \bigg[&  
    \mathbbm{1}\big[P(a^q \mid Immigrants) \geq P(a^q \mid Locals)\big] \beta_{q,a}  \\ 
    - & \mathbbm{1}\big[P(a^q \mid Immigrants) < P(a^q \mid Locals)\big] \beta_{q,a} 
\bigg],
\ee
where \( \mathbbm{1}[\cdot] \) is an indicator function.  

The MTVD is straightforward to estimate from the data and summarizes how migration age changes the distance between immigrants’ and locals’ response distributions. For each response option $a^q$, the term compares the baseline shares $P(a^q\mid Immigrants)$ and $P(a^q\mid Locals)$, and then weights the migration-age effect $\beta_{q,a}$ with a sign chosen so that MTVD increases when migration age amplifies immigrant–local differences and decreases when it shrinks them. Concretely, if immigrants are initially at least as likely as locals to choose $a^q$ $\big(P(a^q\mid Immigrants)\ge P(a^q\mid Locals)\big)$, then $\beta_{q,a}>0$ means migration age widens the gap for that response, while $\beta_{q,a}<0$ means it narrows the gap. If immigrants are initially less likely than locals to choose $a^q$, the sign flips: $\beta_{q,a}>0$ reduces the immigrant–local gap (immigrants “catch up”), whereas $\beta_{q,a}<0$ increases it.

We estimate both these measures using their empirical counterparts. 



\section{Results}\label{sec:results}
In this section, we present the main results, followed by robustness checks and an exploration of heterogeneity across countries of origin.

Table \ref{tab:main_results_table} summarizes the effect of migration age on agreement with four statements about gender role attitudes. The first row reports the effect of migration age on the average score for each statement, where responses are numerically coded as: 1 (Strongly Agree), 2 (Agree), 3 (Neither Agree nor Disagree), 4 (Disagree), and 5 (Strongly Disagree). Lower average scores indicate greater agreement with the statement. The results show that an increase in migration age is associated with stronger agreement with the first statement ("A pre-school child is likely to suffer if his or her mother works") and the second statement ("All in all, family life suffers when the woman has a full-time job"). The strongest effect of migration age, however, is observed for the fourth statement ("A husband's job is to earn money, a wife's job is to look after the home and family"), which most explicitly reflects traditional gender attitudes. A comparison of raw means between immigrants and locals indicates that higher migration ages are associated with attitudes that diverge further from those of locals for Questions 1, 2, and 4. This pattern suggests that earlier migration is linked to greater alignment with local views on traditional gender roles. In contrast, for Question 3 (“Husband and Wife Should Contribute to HH Income”), immigrants’ and locals’ mean scores are already very similar, and the estimated migration-age effect on the average score is small and statistically indistinguishable from zero. Accordingly, for this item there is essentially no meaningful divergence as migration age increases.

The second row of Table \ref{tab:main_results_table} examines the effect of migration age using an alternative scoring function, defining agreement as responses of "Strongly Agree," "Agree," or "Neither Agree nor Disagree." This analysis indicates that a one-year increase in migration age raises the likelihood of agreement with the statements by approximately 1-2 percentage points for three of the four statements. Overall, the results consistently show that higher migration ages are associated with attitudes further removed from those of the local population, reinforcing the importance of early migration for cultural assimilation.

\begin{table}[!h]
    \centering
    \resizebox{\textwidth}{!}{
            \footnotesize
    \begin{tabular}{lcccc}
\toprule \toprule
\textbf{} & \textbf{(1)} & \textbf{(2)} & \textbf{(3)} & \textbf{(4)} \\
\hline\textbf{} & \makecell{\textbf{Pre-School Suffers} \\ \textbf{if Mother Works}} & \makecell{\textbf{Family Suffers} \\ \textbf{if Mother Works} \\ \textbf{Full Time}} & \makecell{\textbf{Husband and Wife} \\ \textbf{Should Contribute} \\ \textbf{to HH Income}} & \makecell{\textbf{Husband Should Earn,} \\ \textbf{Wife Stay at Home}} \\
\hline
Average Score  & -.023 & -.025 & .014 & -.041 \\
\quad  &  (.014) &  (.012) &  (.012) &  (.014) \\
\quad Migrants Average & 2.988 & 3.268 & 2.070 & 3.558 \\
\quad  & [1.062] & [1.141] & [0.923] & [1.188] \\
\quad Locals Average & 3.334 & 3.620 & 2.132 & 3.923 \\
\quad  & [1.006] & [1.049] & [0.899] & [1.036] \\
Binarized & .012 & .011 & -.005 & .017 \\
\quad &  (.006) &  (.006) &  (.006) &  (.005) \\
\quad Migrants Average & 0.682 & 0.569 & 0.941 & 0.468 \\
\quad  & [0.466] & [0.495] & [0.236] & [0.499] \\
\quad Locals Average & 0.560 & 0.433 & 0.943 & 0.323 \\
\quad  & [0.496] & [0.496] & [0.232] & [0.468] \\
\midrule
Agree & .007 & .007 & -.006 & .01 \\
\quad  &  (.006) &  (.005) &  (.003) &  (.006) \\
Neither Agree/Disagree & .005 & .005 & 0 & .007 \\
\quad &  (.007) &  (.006) &  (.005) &  (.007) \\
Disagree & -.012 & -.011 & .005 & -.017 \\
\quad &  (.006) &  (.005) &  (.003) &  (.006) \\
\bottomrule
\end{tabular}

    }
    \caption{Mean Results}\label{tab:main_results_table}
        \begin{minipage}{\linewidth}
    \footnotesize
    \justifying
    \textit{Note:} Regression-style summary results by question (columns (1)--(4), with the same four statements as Table 1). The rows labeled Average Score, Binarized, and the category-specific rows report coefficient estimates, with standard errors clustered at the family level in parentheses. For the rows labeled Immigrants Average and Locals Average, the table reports group means, with corresponding standard deviations in square brackets. Responses are coded as 1 (Strongly Agree), 2 (Agree), 3 (Neither Agree nor Disagree), 4 (Disagree), and 5 (Strongly Disagree). For the category-specific rows, we collapse the five-point scale into three mutually exclusive categories: \emph{Agree} equals $\mathbbm{1}\{\text{Strongly Agree or Agree}\}$, \emph{Neither} equals $\mathbbm{1}\{\text{Neither Agree nor Disagree}\}$, and \emph{Disagree} equals $\mathbbm{1}\{\text{Disagree or Strongly Disagree}\}$.
    \end{minipage}
\end{table}

Table \ref{table:robustness} in the Appendix evaluates the robustness of the results in table \ref{tab:main_results_table}. The first row, presenting the average linear mean and agreement outcomes, replicates the findings from Table \ref{tab:main_results_table}. As discussed in Section \ref{sec:empircalStragey}, a key threat to our identification strategy is the potential correlation between unobserved characteristics that affect gender role attitudes and migration age within families. For example, older siblings may be more likely to hold specific worldviews, or age at the time of the interview may independently influence attitudes toward gender roles. Additionally, the sex of respondents may correlate with birth order, potentially introducing bias. 

To address these concerns, row 2 of Table \ref{table:robustness} includes indicators for whether the respondent is the oldest child in their family and for the respondent’s sex. The results remain consistent with those in Table \ref{tab:main_results_table}, with no meaningful change in the size or direction of the estimated effects. We then add a linear control for the respondent’s age. The estimated effects of migration age on agreement with the second and fourth statements remain robust, but the effect on the first statement (“A pre-school child is likely to suffer if his or her mother works”) becomes statistically insignificant and moves toward zero, suggesting that age-related differences may partly account for the baseline association for question 1. Finally, we replace the linear age control with non-parametric age controls to allow for flexible age-specific effects. The results closely mirror those obtained with the linear specification, reinforcing the overall robustness of our findings.\footnote{With family fixed effects, all coefficients are identified from \emph{within-family} contrasts. When siblings share a single migration episode (common arrival year $A_f$), age at arrival satisfies $MigAge_{if}=A_f-B_{if}$ and is therefore mechanically linked to siblings’ birth years. The “oldest” control captures a non-linear firstborn level difference (oldest vs.\ all others). In two-sibling families that migrate together, the within-family component of $MigAge_{if}$ is proportional to the demeaned “oldest” dummy, so the
two effects cannot be separated within such families; separate identification of the oldest effect comes from families with three or more siblings and/or from families with non-synchronous migration, and in two-sibling families from cross-family variation in siblings’ age gaps. Likewise, controlling linearly for age at interview is redundant when siblings are interviewed in the same wave/year (since within family $\widetilde{Age}_{if}=
\widetilde{MigAge}_{if}$); when interview years differ across siblings, separating age from $MigAge$ uses the resulting within-family variation in interview timing (equivalently, years since arrival).}

As discussed in Section \ref{sec:totalVariationWCAD}, the choice of scoring function may influence conclusions. To address this, we explore the impact of migration age on the Total Variation distance, as discussed in section \ref{sec:convergenceMeasure}. We first start by checking whether the linear probability model yields sensible results. These checks are not intended as empirical tests of the behavioural model; rather, they verify that the estimated system respects the adding-up identities implied by estimating category-specific dummy outcomes that sum to one. In particular, we consider two internal consistency checks. First, for each question $q$, changes in response probabilities across answer categories $a$ must sum to zero, since probabilities must continue to add up to one. That is, we expect
\[\sum_{a}\beta_{q,a}=0 \quad \text{for each } q \].
Second, the baseline family fixed effects should define a valid probability distribution over responses for each question, so we expect
\[
\sum_{a} \gamma^{q,a}_{f(i)}=1 \quad \text{for each } (i,q).
\]
Figure \ref{fig:sumOfFixedEffect} reports the bootstrap distribution of $\sum_a \gamma^{q,a}_{f(i)}$ (aggregated across $i$ and $q$) from Model \ref{eq:linearProbModel}. The distribution is tightly concentrated around $1$, as expected, confirming that the implied counterfactual distribution at $MigAge=0$ is a proper probability distribution. Similarly, Figure \ref{fig:sumOfBetas} shows the bootstrap distribution of $\sum_a \beta_{q,a}$, which is tightly centered at zero. This indicates that increases in some response probabilities are offset by decreases in others, so the resulting distributions remain valid as migration age varies.

\begin{figure}[!h]
    \centering    \includegraphics[width=\textwidth,height=0.8\textheight,keepaspectratio]{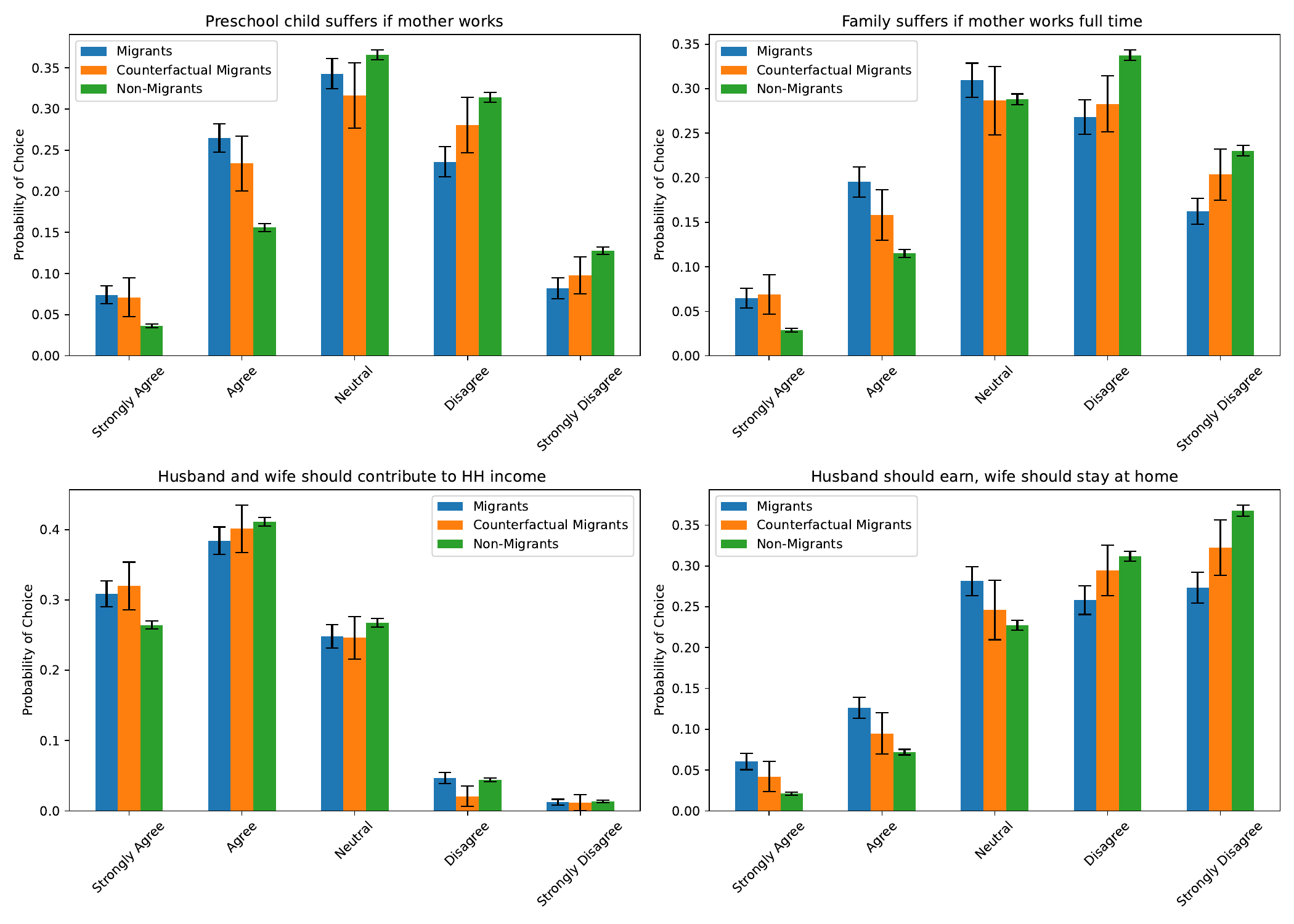}
    \caption{Answer Distributions Across Immigrants and Locals}
    \label{fig:answerDistribution}
    \begin{minipage}{\linewidth}
    \footnotesize
    \justifying
    \textit{Note:} Four-panel bar chart of response distributions across the five response categories (Strongly Agree, Agree, Neutral, Disagree, Strongly Disagree) for each attitude statement. The y-axis reports the “Probability of Choice.” Within each response category, bars compare (i) immigrants (immigrant siblings sample), (ii) “Counterfactual Migrants” constructed by fixing \(MigAge=0\) (arrival at birth) within the estimation framework, and (iii) non-migrants/locals. Error bars denote bootstrap 95\% confidence intervals based on 1,000 replications.
    \end{minipage}
\end{figure}

We now turn to describe our main results. We begin with our “global” measure, which asks how the TV distance would change under a counterfactual in which all immigrant children were effectively born into the same families, but in the UK rather than abroad—that is, we “fix” $MigAge = 0$. The orange bars in Figure \ref{fig:answerDistribution} show the response distribution in this counterfactual scenario. Under this counterfactual, immigrants’ choice probabilities move closer to those of the local UK population than in the observed data, suggesting that spending a larger share of one’s formative years in the UK leads to gender-role attitudes that more closely resemble those of locals.

Table \ref{tab:TotalVariationMeasure} expands on the graphical results and summarizes cultural distance between immigrants and locals using the TV distance. As shown in Section \ref{sec:convergenceMeasure}, TV is numerically equal to the Worst-Case Average Difference all bounded scoring functions \(h:A^q\to[0,1]\). Put differently, a TV distance of \(0.15\) means that—under the scoring rule that makes the two groups look as different as possible—their mean responses can differ by as much as 0.15. TV also has a transparent “mass-moving” interpretation: it equals the minimum share of immigrant responses that would need to be reallocated across categories to exactly match the locals’ response distribution.

The first row shows that the observed immigrant–local distance ranges from \(0.047\) to \(0.148\) across the four statements. In “share-of-responses” terms, this implies that aligning immigrants’ distribution with locals’ would require reassigning about \(4.7\%\) of responses for the income-contribution statement (Question 3), but roughly \(14\%{-}15\%\) for the two work–family statements (Questions 1–2) and for the most explicitly traditional gender-role statement (Question 4). This pattern matches the mean-based results: Question 3 already exhibits near-alignment in the observed data, while the largest gaps appear for statements directly invoking maternal employment and household specialization.

Rows 2–3 compare this observed distance to a counterfactual in which all immigrant children are assigned \(MigAge=0\) (the “born in the UK” benchmark). For Questions 1, 2, and 4, setting migration age to zero moves immigrants closer to locals: TV falls from \(0.146\) to \(0.112\) for Question 1 (a reduction of \(0.034\)), from \(0.137\) to \(0.083\) for Question 2 (a reduction of \(0.054\)), and most sharply from \(0.148\) to \(0.063\) for Question 4 (a reduction of \(0.085\), i.e., about a 57\% decline relative to the baseline gap). Interpreted through the coupling lens, the Question 4 counterfactual implies that the minimum fraction of immigrant responses that must “switch categories” to match locals drops from about \(14.8\%\) to \(6.3\%\). By contrast, for Question 3 the counterfactual change is small and slightly negative (\(-0.008\)), consistent with the idea that there is little room for convergence on this dimension because immigrants and locals already look similar in the observed distribution.

Finally, the last row reports the Marginal Total Variation Divergence, which captures how a small uniform increase in migration age would change cultural distance. Across questions, MTVD is positive, indicating that delaying migration tends to widen the immigrant–local gap: a one-year uniform increase in \(MigAge\) raises TV by about \(0.014\) for Question 1, \(0.022\) for Question 2, and \(0.034\) for Question 4. In “share” terms, these correspond to roughly $1.4$, $2.2$, and $3.4$ additional percentage points of responses that would need to be reallocated to restore parity—again with the largest marginal effect for the statement most directly about traditional household roles. 

In Appendix \ref{app:Ordinal}, we examine a similar analysis, but only restricting the set of ranking function to be those who respect the ordinal structure of the response options (i.e., “strongly agree” must be ranked above “agree,” and so forth). The results we get are similar, suggesting that our upper bounds results we get from the TV measure are not driven by unreasonable orderings of the response categories.

\begin{table}[H]
    \centering
    \resizebox{\textwidth}{!}{
    \footnotesize
    
\begin{tabular}{@{}l c c c c@{}}
\toprule
\textbf{} & \textbf{(1)} & \textbf{(2)} & \textbf{(3)} & \textbf{(4)} \\
\midrule
  & \textbf{\textbf{Preschool child}} & \textbf{\textbf{Family suffers if}} & \textbf{\textbf{Husband and wife}} & \textbf{\textbf{Husband should earn,}} \\
  & \textbf{\textbf{suffers if mother}} & \textbf{\textbf{mother works}} & \textbf{\textbf{should contribute}} & \textbf{\textbf{wife should}} \\
  & \textbf{\textbf{works}} & \textbf{\textbf{full time}} & \textbf{\textbf{to HH income}} & \textbf{\textbf{stay at home}} \\
\midrule
Total Variation - Immigrants Vs. Locals & 0.146 & 0.137 & 0.047 & 0.148 \\
 & [0.111,0.181] & [0.105,0.175] & [0.026,0.081] & [0.112,0.184] \\
Total Variation - Immigrants CF Vs. Locals & 0.112 & 0.083 & 0.056 & 0.063 \\
 & [0.062,0.185] & [0.051,0.162] & [0.03,0.122] & [0.037,0.135] \\
Difference & 0.034 & 0.054 & -0.008 & 0.085 \\
 & [-0.028,0.078] & [-0.013,0.083] & [-0.07,0.023] & [0.018,0.116] \\
\midrule
MTVD & 0.014 & 0.022 & 0.006 & 0.034 \\
 & [-0.008,0.04] & [-0.001,0.044] & [-0.021,0.024] & [0.012,0.056] \\
\bottomrule
\end{tabular}

    }
    \caption{Main Results}
    \label{tab:TotalVariationMeasure}
    \begin{minipage}{\linewidth}
    \footnotesize
    \justifying
    \textit{Note:} Total-variation-based distributional comparisons by question (columns (1)–(4), corresponding to the four statements). The first row reports the total variation distance between the response distributions of immigrants and locals. The second row reports the total variation distance between the counterfactual immigrant response distribution—if all immigrants had immigrated at age 0—and the local response distribution. The “Difference” row reports the difference between the first two rows. Finally, the last row reports the MTVD. Each cell reports a point estimate, followed by bootstrap 95\% confidence intervals, in square brackets, based on 1,000 replications.
    \end{minipage}
\end{table}

\begin{table}[!h]
    \centering
    \resizebox{\textwidth}{!}{
    \footnotesize
    
\begin{tabular}{@{}l c c c c@{}}
\toprule
  & Preschool child & Family suffers if & Husband and wife & Husband should earn, \\
  &  suffers if mother &  mother works &  should contribute &  wife should \\
  &  works &  full time  &  to HH income  &  stay at home \\\\
\midrule
Total Variation - Western Europe, USA and Canada Vs. Locals & 0.08 & 0.14 & 0.08 & 0.095 \\
 & [0.035, 0.18] & [0.084, 0.224] & [0.034, 0.171] & [0.046, 0.189] \\
 Total Variation - CF Western Europe, USA and Canada Vs. Locals & 0.096 & 0.128 & 0.137 & 0.107 \\
 & [0.05, 0.242] & [0.07, 0.274] & [0.066, 0.309] & [0.061, 0.276] \\
Difference & -0.016 & 0.013 & -0.057 & -0.012 \\
 & [-0.118, 0.046] & [-0.096, 0.073] & [-0.197, 0.038] & [-0.16, 0.043] \\
\midrule
Total Variation - Rest of the world Vs. Locals & 0.186 & 0.198 & 0.04 & 0.201 \\
 & [0.147, 0.231] & [0.15, 0.249] & [0.016, 0.088] & [0.151, 0.25] \\
 Total Variation - CF Rest of the world Vs. Locals & 0.15 & 0.127 & 0.042 & 0.1 \\
 & [0.078, 0.239] & [0.077, 0.211] & [0.026, 0.134] & [0.045, 0.178] \\
Difference & 0.036 & 0.072 & -0.002 & 0.101 \\
 & [-0.042, 0.103] & [0.004, 0.116] & [-0.082, 0.024] & [0.031, 0.146] \\
\midrule
MTVD Western Europe, USA and Canada Vs. Locals & -0.008 & 0.011 & -0.01 & 0.002 \\
 & [-0.056, 0.05] & [-0.041, 0.076] & [-0.088, 0.063] & [-0.057, 0.071] \\
\midrule
MTVD Rest of the world Vs. Locals & 0.014 & 0.027 & 0.006 & 0.039 \\
 & [-0.014, 0.047] & [0.006, 0.048] & [-0.023, 0.031] & [0.014, 0.064] \\
\midrule
\bottomrule
\end{tabular}

    }
    \caption{Heterogeneity across Country of Origin}
    \label{tab:TV_hetro_WesternVs}
        \begin{minipage}{\linewidth}
    \footnotesize
    \justifying
    \textit{Note:} This table is analogous to Table \ref{tab:TotalVariationMeasure}. In Section 1, we compare the total variation distance between the response distributions of immigrants from Western Europe, the United States, and Canada and those of locals. The next row reports the corresponding counterfactual for these immigrants under the scenario in which they were born in the UK. The final row in Section 1 reports the difference between the observed and counterfactual distances. Section 2 presents the same set of quantities for immigrants from the rest of the world. The last section reports the MTVD for immigrants from Western Europe, the United States, and Canada and, separately, for immigrants from the rest of the world. Each cell reports a point estimate, with a 95\% confidence interval in square brackets constructed from 1{,}000 bootstrap replications.
    \end{minipage}
\end{table}




Finally, we examine whether the convergence patterns vary systematically with origin-country context. Table \ref{tab:TV_hetro_WesternVs} splits immigrant siblings into those from Western Europe / the US / Canada versus the rest of the world, and reports the same TV distance as in the main results. The contrast between the two groups is sharp. For immigrants from Western Europe / the US / Canada, the immigrant–local distance is already modest—TV ranges from about 0.08 to 0.14 across questions. Moreover, shifting this group to the \(MigAge=0\) counterfactual produces changes that are economically small and statistically indistinguishable from zero (e.g., differences of \(-0.016\), \(0.013\), \(-0.057\), and \(-0.012\)). Consistent with this, the MTVD estimates for the Western group are essentially zero (e.g., \(-0.008\), \(0.011\), \(-0.010\), \(0.002\)), indicating little marginal sensitivity of cultural distance to migration timing within this origin set.

By contrast, immigrants from the rest of the world exhibit substantially larger gaps relative to locals for the statements most explicitly tied to gender roles. In the observed data, TV is about 0.186, 0.198, and 0.201 for Questions 1, 2, and 4, respectively—i.e., roughly 19–20\% of responses would need to be reassigned across categories to match locals’ distributions. Fixing \(MigAge=0\) reduces these distances, with especially large improvements for the “family suffers if mother works full time” statement (a reduction of $0.072$) and the “husband should earn, wife should stay at home” statement (a reduction of $0.101$). The marginal results tell the same story: MTVD is positive for this group, and is statistically distinguishable from zero for Questions 2 and 4 (e.g., $0.027$) for Question 2, and $0.039$ for Question 4). Interpreted in “share-moving” terms, a one-year uniform delay in migration increases the minimum mass that would need to be reassigned by about 2.7 pp (Q2) and 3.9 pp (Q4) for non-Western-origin immigrants.

Appendix Tables \ref{tab:hetro} and \ref{tab:TV_hetro_differentUK} refines this heterogeneity by measuring cultural proximity to the UK using the World Values Survey: for each origin country, we construct a gender-attitude profile from several WVS gender-related items and split countries at the median distance from the UK. Immigrants from “similar-to-UK” origins have relatively small baseline TV gaps ($0.114$, $0.059$, $0.06$, $0.042$ across the four questions). In contrast, immigrants from “different-from-UK” origins exhibit much larger baseline gaps for Questions 1, 2, 3 and especially 4 ( $0.231$, $0.293$, $0.059$ and $0.309$), implying that roughly $23–27\%$ of responses would need to be reallocated to match locals’ distributions on those dimensions.

Crucially, it is this culturally distant group for which migration timing meaningfully predicts convergence. For the most explicit gender-role statement (Question 4), fixing \(MigAge=0\) reduces TV from $0.309$ to $0.207$, a decline of $0.101$. The MTVD for the same question is $0.05$, implying that each additional year of migration age increases the “minimum mass to move” by roughly $5$ percentage points for immigrants from culturally distant origins. By comparison, the culturally similar group shows no statistically detectable MTVD effects, estimates are close to zero and statistically indistinguishable from zero across all questions.

Taken together, Tables \ref{tab:TV_hetro_WesternVs} and \ref{tab:TV_hetro_differentUK} point to a consistent mechanism: migration age matters most when the origin-country norm environment is far from the UK’s. When baseline cultural distance is small, there is limited scope for convergence and little detectable sensitivity to migration timing. When baseline distance is large, earlier arrival substantially compresses the attitude gap—especially for statements directly about traditional household specialization—supporting an assimilation/exposure channel in which formative-time exposure to host-country norms plays a central role.

\section{Conclusions}\label{sec:conclusions}
This paper studies cultural assimilation through the lens of gender-role attitudes, asking whether \emph{age at arrival} causally shapes how closely immigrants’ beliefs align with those of the UK-born population. Motivated by the idea that childhood is a formative period for socialization, we combine a sibling fixed-effects design with a distributional perspective on convergence in survey responses. The sibling design isolates within-family variation in age at arrival—holding fixed shared parental background, origin, and the migration episode—so that differences across siblings primarily capture differences in exposure to UK norms during childhood.

Our results point to a simple takeaway: shifting the arrival-age distribution earlier—so immigrant children spend a larger share of childhood in the UK—would make their gender-role attitudes more similar to those of UK-born locals. This pattern is consistent with a socialization interpretation of assimilation: exposure during formative years appears to shift attitudes toward the host-country norm, while arriving later—after attitudes have had more time to form in the origin environment—leaves immigrants’ beliefs further from the local baseline. In this sense, age at arrival operates as a proxy for the timing and intensity of cultural exposure, and our estimates suggest that the “window” in which that exposure matters is concentrated in childhood.
We also find meaningful heterogeneity across origin contexts that further supports this interpretation. Immigrants from culturally closer origins—such as Western Europe, the US, or Canada, and more generally origins with gender-role norms closer to the UK as measured in the World Values Survey—exhibit smaller baseline differences from locals and weaker sensitivity to age at arrival. By contrast, immigrants from culturally more distant origins display larger baseline gaps and substantially larger age-at-arrival effects: when the origin-country norm is further from the UK baseline, arriving earlier matters more for convergence. This pattern is difficult to reconcile with purely mechanical explanations and is instead consistent with the idea that the key margin is the intensity and timing of exposure to the host-country cultural environment.

Our main outcome of interest is respondents’ answers to questions about gender norms, rather than directly observed actions. While economists often prefer revealed behavior over stated preferences, we argue that in the domain of cultural conflict and social integration, stated views are themselves a first-order object of interest. The attitudes immigrants and locals express—in surveys, classrooms, workplaces, and public discourse—are precisely the differences that can be noticed by others, become socially salient, and potentially generate stigma, disdain, or interpersonal tension, even when day-to-day behavior inside the home is harder to observe. That said, in our data, stated attitudes are meaningfully correlated with household behavior, including measures of within-household specialization and related patterns of time use and task allocation.

Beyond the substantive findings, the paper contributes a measurement approach for comparing attitudes in ordinal surveys. Because Likert responses have no canonical numerical scale, mean differences implicitly embed arbitrary choices about how far apart categories are and which parts of the distribution matter. Total Variation sidesteps this problem by treating responses as distributions: it equals the largest possible difference in average responses across all bounded recodings of the categories, so evidence of convergence is robust to how one “scores” the answers. At the same time, TV has a concrete interpretation as the minimal fraction of responses that would need to change for immigrants’ distribution to match locals’. This combination—robustness to scoring and a transparent “share who would need to switch” interpretation—makes TV a useful reporting device for cultural convergence and, more broadly, for any setting where researchers want to summarize treatment effects on survey distributions rather than only on means.

\newpage
\bibliographystyle{apalike} 
\bibliography{sample}

\newpage
\appendix
\section{Measuring Gender Attitude Similarities between Countries and the UK}\label{appendix:similarity}

To examine heterogeneity in responses by country of origin, we complement the UKHLS data with Wave 7 of the World Values Survey (WVS), conducted between 2017-2020. We selected this wave for its comprehensive country coverage and question set. The WVS questionnaire presents statements that respondents rate on a four-point scale. To calculate country-level gender attitude scores, average responses to the following statements:
\begin{itemize}
    \item When a mother works for pay, the children suffer
    \item On the whole, men make better political leaders than women do
    \item A university education is more important for a boy than for a girl
    \item On the whole, men make better business executives than women do
    \item Being a housewife is just as fulfilling as working for pay
    \item When jobs are scarce, men should have more right to a job than women
    \item If a woman earns more money than her husband, it's almost certain to cause problems
\end{itemize}

To measure the similarity between the gender attitudes of an origin country and those of the UK, we compute the distance between each country's seven-component vector of average responses and the corresponding UK vector. Specifically, let $\bar{a}_q^c$ represent the average response to statement $q$ by respondents from country $c$. We then calculate the similarity in gender attitudes between country $c$ and the UK using the formula:

\[
\text{Diff}_{c,UK}  = \sqrt{\sum_{q} (\bar{a}_{q}^c - \bar{a}_{q}^{UK})^2}.
\]

This metric intuitively captures how closely the attitudes of each country align with those of the UK. In our analysis in Section \ref{sec:results}, we split the sample based on the median distance of each country of origin in our dataset from the UK.

\section{Ordinal Robust Measure of Convergence}\label{app:Ordinal}

Our main analysis measures convergence in gender-role attitudes using the total variation (TV) distance between the response distributions of immigrants and locals. TV is attractive because it is numerically equal to the worst-case average difference over all bounded scoring rules. At the same time, TV treats response categories as an unordered finite set. For order-type items such as ``Strongly Agree'' through ``Strongly Disagree,'' it is also natural to consider a robustness concept that respects order. 

This appendix estimate such a measure. We restrict attention to bounded scoring rules that are \emph{monotone} in the ordered response categories. The resulting worst-case average difference coincides with the Kolmogorov distance between the two cumulative response distributions. This yields an ordinally meaningful analogue to the TV-based analysis in the main text.

\subsection{Setup and Definition}

Let the ordinal response to question $q$ take values in the ordered support
\[
A^q \in \{1,2,\dots,K_q\}, \qquad 1 \prec 2 \prec \cdots \prec K_q.
\]
For any two distributions $P$ and $Q$ on $\{1,\dots,K_q\}$, define the cumulative distribution functions
\[
F_P(k) := P(A^q \le k), \qquad F_Q(k) := Q(A^q \le k), \qquad k=1,\dots,K_q.
\]

Define the order-restricted worst-case average difference by
\[
\mathrm{OWCAD}(P,Q)
:=
\sup_{\substack{h:\{1,\dots,K_q\}\to[0,1]\\ h \text{ nondecreasing}}}
\left| \mathbb E_P[h(A^q)] - \mathbb E_Q[h(A^q)] \right|.
\]
Relative to TV, the function class is smaller: we still allow arbitrary bounded recodings of the response categories, but only among recodings that preserve the ordinal ranking of the answers.

Next, we show that OWCAD is equivariant with respect to the Kolmogorov distance between two distributions. Although we believe this result is well known, we were unable to find a direct reference in the literature, so we provide a simple proof here.

\begin{lemma}[OWCAD equals the Kolmogorov distance on an ordinal support]
For ordered categorical outcomes,
\[
\mathrm{OWCAD}(P,Q)
\;=\;
\max_{k\in\{1,\dots,K_q-1\}} |F_P(k)-F_Q(k)|.
\]
\end{lemma}

\begin{proof}
Let $h$ be any nondecreasing function from $\{1,\dots,K_q\}$ to $[0,1]$, and define increments
\[
\Delta_k := h(k+1)-h(k), \qquad k=1,\dots,K_q-1.
\]
Monotonicity implies $\Delta_k\ge 0$, and boundedness implies
\[
\sum_{k=1}^{K_q-1}\Delta_k = h(K_q)-h(1)\le 1.
\]
Writing
\[
h(a)=h(1)+\sum_{k=1}^{a-1}\Delta_k,
\]
and taking expectations under $P$ gives
\[
\mathbb E_P[h(A^q)]
=
h(1)+\sum_{k=1}^{K_q-1}\Delta_k P(A^q>k).
\]
The same argument for $Q$ yields
\[
\mathbb E_Q[h(A^q)]
=
h(1)+\sum_{k=1}^{K_q-1}\Delta_k Q(A^q>k).
\]
Subtracting and using $P(A^q>k)=1-F_P(k)$ and $Q(A^q>k)=1-F_Q(k)$,
\[
\mathbb E_P[h(A^q)]-\mathbb E_Q[h(A^q)]
=
-\sum_{k=1}^{K_q-1}\Delta_k\big(F_P(k)-F_Q(k)\big).
\]
Hence
\begin{align*}
\left| \mathbb E_P[h(A^q)]-\mathbb E_Q[h(A^q)] \right|
&\le
\sum_{k=1}^{K_q-1}\Delta_k |F_P(k)-F_Q(k)| \\
&\le
\left(\sum_{k=1}^{K_q-1}\Delta_k\right)
\max_k |F_P(k)-F_Q(k)| \\
&\le
\max_k |F_P(k)-F_Q(k)|.
\end{align*}
This proves the upper bound. For the reverse direction, choose
\[
k^\star \in \arg\max_{k\in\{1,\dots,K_q-1\}} |F_P(k)-F_Q(k)|
\]
and set $h_{k^\star}(a)=\mathbf{1}\{a>k^\star\}$. Then
\[
\mathbb E_P[h_{k^\star}(A^q)]-\mathbb E_Q[h_{k^\star}(A^q)]
=
P(A^q>k^\star)-Q(A^q>k^\star)
=
-(F_P(k^\star)-F_Q(k^\star)),
\]
so
\[
\left| \mathbb E_P[h_{k^\star}(A^q)]-\mathbb E_Q[h_{k^\star}(A^q)] \right|
=
|F_P(k^\star)-F_Q(k^\star)|.
\]
Therefore the upper bound is attained, proving the claim.
\end{proof}

Lemma 1 shows that the ordinally robust worst-case difference is exactly the Kolmogorov distance
\[
KD(P,Q) := \max_{k\in\{1,\dots,K_q-1\}} |F_P(k)-F_Q(k)|.
\]
This makes the interpretation especially transparent: the relevant discrepancy is the largest gap between the two cumulative response distributions over all monotone thresholds.

\subsection{Counterfactual and Marginal Kolmogorov-Distance Estimands}
Here we follow similar logic as section \ref{sec:totalVariationWCAD} in the main text. Let $p^q_{Locals}$ denote the observed response distribution of UK-born locals for question $q$, and let $p^q_{Immigrants}$ denote the corresponding observed immigrant response distribution. Define the global counterfactual distribution for immigrants under $MigAge=0$ by
\[
p^{CF,0}_{Immigrants}(a)
=
\int
P\big(A^q_{i,MigAge=0}=a \mid MigAge=x, f(i), Immigrants\big)
\,
dF\big(f(i), MigAge=x \mid Immigrants\big).
\]
The associated global Kolmogorov-distance effect is
\[
\Delta KD_q^0
:=
KD\!\left(p^q_{Locals}, p^q_{Immigrants}\right)
-
KD\!\left(p^q_{Locals}, p^{CF,0}_{Immigrants}\right).
\]
A positive value of $\Delta KD_q^0$ means that setting immigrant migration age to zero reduces the maximum cumulative immigrant--local discrepancy for question $q$.

To define the marginal analogue, let $p^q_{\epsilon,Immigrants}$ denote the immigrant response distribution under a small uniform perturbation $\epsilon$ to migration age:
\[
p^q_{\epsilon,Immigrants}(a)
=
\int
P\big(A^q_{i,x+\epsilon}=a \mid MigAge=x, f(i), Immigrants\big)
\,
dF\big(f(i), MigAge=x \mid Immigrants\big).
\]
We then define the \emph{Marginal Kolmogorov Divergence} by
\[
MKD_q
:=
\frac{\partial}{\partial \epsilon}
KD\!\left(p^q_{\epsilon,Immigrants}, p^q_{Locals}\right)
\bigg|_{\epsilon=0}.
\]
This estimand captures how a small uniform increase in migration age changes the largest immigrant--local cumulative gap for question $q$.

\subsection{Estimation}

As in the main text, we estimate a family fixed-effects linear probability model separately for each response category $a$ of question $q$:
\begin{equation}
\Pr(A_i^q=a \mid MigAge_i, f(i))
=
MigAge_i \beta_{q,a} + \gamma_{f(i)}^{q,a}.
\label{eq:kd_lpm}
\end{equation}
Here, again, $\beta_{q,a}$ is the effect of migration age on the probability of response $a$, and $\gamma_{f(i)}^{q,a}$ is the family fixed effect for category $a$.

Let $\widehat p^q_{Immigrants}(a)$ denote the observed immigrant share choosing category $a$, $\widehat p^q_{Locals}(a)$ the corresponding local share, and $\widehat p^{CF,0,q}_{Immigrants}(a)$ the counterfactual immigrant share obtained by setting $MigAge_i=0$ in the estimated model for all immigrants and averaging the predicted probabilities.

Define the empirical cumulative distributions
\[
\widehat F^q_{Immigrants}(k)
=
\sum_{a\le k}\widehat p^q_{Immigrants}(a),\qquad
\widehat F^q_{Locals}(k)
=
\sum_{a\le k}\widehat p^q_{Locals}(a),\qquad
\widehat F^{CF,0,q}_{Immigrants}(k)
=
\sum_{a\le k}\widehat p^{CF,0,q}_{Immigrants}(a).
\]
Then the empirical Kolmogorov distances are given by 
\[
\widehat{KD}^q_{Immigrants,Locals}
=
\max_{k\in\{1,\dots,K_q-1\}}
\left|
\widehat F^q_{Immigrants}(k)-\widehat F^q_{Locals}(k)
\right|,
\]
and
\[
\widehat{KD}^{CF,0,q}_{Immigrants,Locals}
=
\max_{k\in\{1,\dots,K_q-1\}}
\left|
\widehat F^{CF,0,q}_{Immigrants}(k)-\widehat F^q_{Locals}(k)
\right|.
\]
Our estimator of the global effect is
\[
\widehat{\Delta KD_q^0}
=
\widehat{KD}^q_{Immigrants,Locals}
-
\widehat{KD}^{CF,0,q}_{Immigrants,Locals}.
\]

To estimate the marginal effect, for each cutoff $k$, define the cumulative immigrant--local gap
\[
G_q(k)
:=
\widehat F^q_{Immigrants}(k)-\widehat F^q_{Locals}(k),
\]
and the cumulative migration-age slope
\[
B_q(k)
:=
\sum_{a\le k}\widehat\beta_{q,a}.
\]
Because
\[
KD_q = \max_k |G_q(k)|,
\]
the derivative depends on the cutoff(s) attaining the maximum. If the maximizing cutoff is unique, say
\[
k_q^\star \in \arg\max_k |G_q(k)|,
\]
then the marginal effect is estimated by\footnote{If several cutoffs tie for the maximum, the derivative of the max operator need not be unique. In that case, the marginal effect is naturally set-valued. Let
\[
\mathcal K_q^\star
:=
\arg\max_k |G_q(k)|.
\]
Then the identified interval is
\[
\widehat{MKD}_q
\in
\left[
\min_{k\in\mathcal K_q^\star}\mathrm{sign}(G_q(k))B_q(k),
\;
\max_{k\in\mathcal K_q^\star}\mathrm{sign}(G_q(k))B_q(k)
\right].
\] 
}
\[
\widehat{MKD}_q
=
\mathrm{sign}\!\big(G_q(k_q^\star)\big)\,B_q(k_q^\star).
\]

Inference follows the same block-bootstrap logic as in the TV analysis. We resample families with replacement, re-estimate the category-specific family fixed-effects linear probability models, reconstruct the counterfactual category shares under $MigAge=0$, and recompute the implied cumulative distributions and Kolmogorov distances in each bootstrap draw. 

\subsection{Empirical Results}

Table \ref{tab:KolmogorovMeasure} reports the Kolmogorov-distance analogues to the TV results in the main text. The first row reports the observed immigrant--local Kolmogorov distance for each question. The second row reports the corresponding counterfactual distance under $MigAge=0$, and the third row reports the difference between the observed and counterfactual distances. The last row reports the marginal effect.

\begin{table}[H]
    \centering
    \resizebox{\textwidth}{!}{
    \footnotesize
    \begin{tabular}{@{}l c c c c@{}}
\toprule
  & Preschool child & Family suffers if & Husband and wife & Husband should earn, \\
  &  suffers if mother &  mother works &  should contribute &  wife should \\
  &  works &  full time  &  to HH income  &  stay at home \\
  \midrule
Kolmogorov - Migrants Vs. Locals & 0.146 & 0.137 & 0.044 & 0.148 \\
 & [0.103, 0.193] & [0.087, 0.195] & [0.023, 0.088] & [0.098, 0.198] \\
Kolmogorov - CF Migrants Vs. Locals & 0.112 & 0.083 & 0.056 & 0.063 \\
 & [0.049, 0.2] & [0.044, 0.186] & [0.027, 0.156] & [0.021, 0.162] \\
 \midrule
Difference & 0.034 & 0.054 & -0.011 & 0.085 \\
 & [-0.043, 0.097] & [-0.036, 0.099] & [-0.104, 0.022] & [-0.012, 0.133] \\
 \midrule
Marginal KS & 0.007 & 0.011 & -0.002 & 0.017 \\
 & [-0.007, 0.024] & [-0.007, 0.027] & [-0.015, 0.017] & [0.001, 0.034] \\
\bottomrule
\end{tabular}

    }
    \caption{Kolmogorov-Distance-Based Distributional Convergence}
    \label{tab:KolmogorovMeasure}
    \begin{minipage}{\linewidth}
    \footnotesize
    \justifying
    \textit{Note:} Kolmogorov-distance-based distributional comparisons by question. The first row reports the observed Kolmogorov distance between the cumulative response distributions of immigrants and locals. The second row reports the corresponding distance between the counterfactual immigrant response distribution under $MigAge=0$ and the local response distribution. The ``Difference'' row reports the reduction in the Kolmogorov distance implied by setting $MigAge=0$. The final row reports the marginal Kolmogorov effect. Each cell reports a point estimate, followed by a bootstrap 95\% confidence interval in square brackets. 
    \end{minipage}
\end{table}

Table \ref{tab:KolmogorovMeasure} reports results for a distance measure that captures the maximum immigrant–local discrepancy across cumulative thresholds of the ordinal response scale. Relative to Total Variation, this measure imposes an additional monotonicity restriction by considering only increasing functions of the response categories. Nevertheless, the results closely resemble those in Table \ref{tab:TotalVariationMeasure}: the observed immigrant–local gap, the counterfactual immigrant–local gap, and the marginal effects all exhibit the same qualitative pattern. The similarity in magnitudes indicates that our baseline TV results are not driven by arbitrary or non-monotone reweighting across categories. Instead, the same convergence pattern emerges even when the comparison is restricted to ordered cumulative contrasts, suggesting that the effect of migration age reflects meaningful shifts in the distribution of attitudes rather than merely diffuse reshuffling across adjacent response categories. 

\section{Tables and Figures}
\begin{table}[H]
    \centering
    \footnotesize
    \begin{tabular}{lcc}
\hline
\textbf{Continent} & \textbf{Immigrants in Siblings Sample} & \textbf{All Immigrants} \\
\hline
Africa & 154 (21\%)& 858 (20)\% \\
America & 40 (5\%)& 503 (12)\% \\
Asia & 350 (48\%)& 1889 (45)\% \\
Europe & 167 (23\%)& 897 (21)\% \\
Other or unknown & 18 (2\%)& 90 (2)\% \\
\hline
\end{tabular}

    \caption{Origin Continent}\label{table:originCont}
        \begin{minipage}{\linewidth}
    \footnotesize
    \justifying
    \textit{Note:} Composition of immigrants origins by continent. Rows list continents (and an “Other or unknown” category). Two columns report counts and percentages for immigrants in the siblings sample and for all immigrants, with percentages shown in parentheses next to the counts.
    \end{minipage}
\end{table}

\begin{table}[H]
    \centering
    \resizebox{\textwidth}{!}{
    \footnotesize
    \begin{tabular}{lcccc}
\hline
\textbf{} & \makecell{\textbf{Pre-School} \\ \textbf{Suffers} \\ \textbf{if Mother Works}} & \makecell{\textbf{Family Suffers} \\ \textbf{if Mother Works} \\ \textbf{Full Time}} & \makecell{\textbf{Husband and Wife} \\ \textbf{Should Contribute} \\ \textbf{to HH Income}} & \makecell{\textbf{Husband Should} \\ \textbf{Earn,} \\ \textbf{Wife Stay at Home}} \\
\hline
Immigrant Siblings & 2.99 & 3.27 & 2.07 & 3.56 \\
  & (1.06) & (1.14) & (0.92) & (1.19) \\
Local Siblings & 3.33 & 3.62 & 2.13 & 3.92 \\
  & (1.01) & (1.05) & (0.90) & (1.04) \\
Immigrants & 2.76 & 2.91 & 2.13 & 3.37 \\
  & (1.10) & (1.15) & (0.89) & (1.18) \\
Locals & 3.16 & 3.26 & 2.29 & 3.70 \\
  & (1.02) & (1.07) & (0.89) & (1.05) \\
\hline
Strongly Agree \\
\hline
Immigrants Siblings & 0.07 & 0.06 & 0.31 & 0.06 \\
 & (0.26) & (0.25) & (0.46) & (0.24) \\
Local Siblings & 0.04 & 0.03 & 0.26 & 0.02 \\
 & (0.19) & (0.17) & (0.44) & (0.15) \\
Immigrants & 0.13 & 0.11 & 0.25 & 0.07 \\
 & (0.33) & (0.31) & (0.43) & (0.26) \\
Locals & 0.05 & 0.05 & 0.19 & 0.03 \\
 & (0.23) & (0.21) & (0.39) & (0.17) \\
\hline
Agree \\
\hline
Immigrants Siblings & 0.26 & 0.20 & 0.38 & 0.13 \\
 & (0.44) & (0.40) & (0.49) & (0.33) \\
Local Siblings & 0.16 & 0.12 & 0.41 & 0.07 \\
 & (0.36) & (0.32) & (0.49) & (0.26) \\
Immigrants & 0.31 & 0.29 & 0.45 & 0.17 \\
 & (0.46) & (0.45) & (0.50) & (0.37) \\
Locals & 0.21 & 0.21 & 0.41 & 0.09 \\
 & (0.40) & (0.40) & (0.49) & (0.29) \\
\hline
Neither Agree or Disagree \\
\hline
Immigrants Siblings & 0.34 & 0.31 & 0.25 & 0.28 \\
 & (0.48) & (0.46) & (0.43) & (0.45) \\
Local Siblings & 0.37 & 0.29 & 0.27 & 0.23 \\
 & (0.48) & (0.45) & (0.44) & (0.42) \\
Immigrants & 0.30 & 0.27 & 0.24 & 0.27 \\
 & (0.46) & (0.45) & (0.43) & (0.44) \\
Locals & 0.36 & 0.32 & 0.32 & 0.28 \\
 & (0.48) & (0.47) & (0.47) & (0.45) \\
\hline
Disagree \\
\hline
Immigrants Siblings & 0.24 & 0.27 & 0.05 & 0.26 \\
 & (0.42) & (0.44) & (0.21) & (0.44) \\
Local Siblings & 0.31 & 0.34 & 0.04 & 0.31 \\
 & (0.46) & (0.47) & (0.20) & (0.46) \\
Immigrants & 0.20 & 0.23 & 0.05 & 0.29 \\
 & (0.40) & (0.42) & (0.22) & (0.46) \\
Locals & 0.29 & 0.30 & 0.06 & 0.34 \\
 & (0.45) & (0.46) & (0.24) & (0.47) \\
\hline
Strongly Disagree \\
\hline
Immigrants Siblings & 0.08 & 0.16 & 0.01 & 0.27 \\
 & (0.28) & (0.37) & (0.11) & (0.45) \\
Local Siblings & 0.13 & 0.23 & 0.01 & 0.36 \\
 & (0.33) & (0.42) & (0.11) & (0.48) \\
Immigrants & 0.06 & 0.09 & 0.01 & 0.19 \\
 & (0.24) & (0.29) & (0.11) & (0.40) \\
Locals & 0.09 & 0.13 & 0.01 & 0.26 \\
 & (0.29) & (0.33) & (0.11) & (0.44) \\
\hline
\end{tabular}

    }
    \caption{Response Distribution By Group}\label{table:responseByGroup}
                \begin{minipage}{\linewidth}
    \footnotesize
    \justifying
    \textit{Note:} Response distributions for each of the four statements (columns), reported for multiple groups (rows): Immigrant Siblings, Local Siblings, Immigrants, and Locals. The top panel reports the mean response level for each group–question pair with the standard deviation in parentheses. Subsequent panels report, for each response category (Strongly Agree, Agree, Neither Agree or Disagree, Disagree, Strongly Disagree), the mean share/probability in that category with the standard deviation in parentheses.
    \end{minipage}
\end{table}

\begin{table}[H]
    \centering
    \resizebox{\textwidth}{!}{
    \footnotesize
    \begin{tabular}{lcccccccccccc}
\hline
\textbf{Group} & \textbf{Min} & \textbf{1p} & \textbf{5p} & \textbf{10p} & \textbf{25p} & \textbf{50p} & \textbf{75p} & \textbf{90p} & \textbf{99p} & \textbf{Max} & \textbf{Mean} \\
\hline
Immigrant-Siblings & 16.00 & 16.00 & 16.00 & 17.00 & 19.00 & 22.00 & 27.00 & 33.00 & 54.00 & 71.00 & 23.95 \\
Local-Siblings & 16.00 & 16.00 & 16.00 & 17.00 & 19.00 & 23.00 & 28.00 & 36.00 & 63.00 & 96.00 & 25.29 \\
Immigrants & 16.00 & 16.00 & 19.00 & 23.00 & 32.00 & 43.00 & 56.00 & 70.00 & 85.00 & 100.00 & 44.95 \\
Locals & 16.00 & 16.00 & 18.00 & 21.00 & 31.00 & 48.00 & 64.00 & 75.00 & 88.00 & 103.00 & 48.32 \\
\hline
\end{tabular}

    }
    \caption{Age Distribution By Group}\label{table:ageByGroup}
    \begin{minipage}{\linewidth}
    \footnotesize
    \justifying
    \textit{Note:} 
    Summary statistics for age by group (Immigrant-Siblings, Local-Siblings, Immigrants, Locals). Columns report the minimum, selected percentiles (1p, 5p, 10p, 25p, 50p, 75p, 90p, 99p), maximum, and the mean age for each group.
    \end{minipage}
\end{table}

\begin{table}[H]
    \centering
    \footnotesize
    \begin{tabular}{lcccc}
\hline
\textbf{} & \makecell{\textbf{Pre-School} \\ \textbf{Suffers} \\ \textbf{if Mother Works}} & \makecell{\textbf{Family Suffers} \\ \textbf{if Mother Works} \\ \textbf{Full Time}} & \makecell{\textbf{Husband and Wife} \\ \textbf{Should Contribute} \\ \textbf{to HH Income}} & \makecell{\textbf{Husband Should} \\ \textbf{Earn,} \\ \textbf{Wife Stay at Home}} \\
\hline
Age at Interview   & 25.2 & 25.2 & 25.2 & 25.2 \\ 
  & (9.1) & (9.2) & (9.2) & (9.1) \\ 
Sex  & 0.52 & 0.52 & 0.52 & 0.52 \\ 
  & (0.50) & (0.50) & (0.50) & (0.50) \\ 
Wave b   & 0.11 & 0.11 & 0.11 & 0.11 \\ 
  & (0.32) & (0.32) & (0.32) & (0.31) \\ 
Wave d   & 0.28 & 0.28 & 0.28 & 0.28 \\ 
   & (0.45) & (0.45) & (0.45) & (0.45) \\ 
Wave j   & 0.61 & 0.61 & 0.61 & 0.61 \\ 
   & (0.49) & (0.49) & (0.49) & (0.49) \\ 
Average Score  & 3.34 & 3.63 & 2.13 & 3.93 \\ 
  & (1.01) & (1.05) & (0.90) & (1.03) \\ 
Family Size  & 2.5   & 2.5   & 2.5   & 2.5   \\ 
      & (0.8)   & (0.8)   & (0.8)   & (0.8)   \\ 
Num. Families  & 2777.0   & 2776.0   & 2775.0   & 2773.0   \\ 
Num. Obs  & 6260.0  & 6259.0  & 6254.0  & 6250.0  \\ 
\bottomrule
\end{tabular}

    \caption{Descriptive Statistics - Local Families}
    \label{desc_locals}
    \begin{minipage}{\linewidth}
    \footnotesize
    \justifying
    \textit{Note:} Descriptive statistics for the local-siblings sample, shown separately for each of the four attitude questions (columns (1)–(4): Pre-School Suffers if Mother Works; Family Suffers if Mother Works Full Time; Husband and Wife Should Contribute to Household Income; Husband Should Earn, Wife Stay at Home). Rows report the mean of key variables (Migration Age, Migration Year, Age at Interview, Sex, survey-wave indicators, Average Score, and Family Size), with standard deviations in parentheses. The bottom rows report the number of families and number of observations used for each column.
    \end{minipage}

\end{table}

\begin{table}[H]    
    \centering
    \resizebox{\textwidth}{!}{
    \footnotesize
    \begin{tabular}{lcccc}
\toprule \toprule
\textbf{} & \textbf{(1)} & \textbf{(2)} & \textbf{(3)} & \textbf{(4)} \\
\hline\textbf{} & \makecell{\textbf{Pre-School Suffers} \\ \textbf{if Mother Works}} & \makecell{\textbf{Family Suffers} \\ \textbf{if Mother Works} \\ \textbf{Full Time}} & \makecell{\textbf{Husband and Wife} \\ \textbf{Should Contribute} \\ \textbf{to HH Income}} & \makecell{\textbf{Husband Should Earn,} \\ \textbf{Wife Stay at Home}} \\
\hline
Average Mean & -0.023 & -0.025 & 0.014 & -0.041 \\
 & (0.014) & (0.012) & (0.012) & (0.014) \\
+ Oldest indicator and sex indicator & -0.024 & -0.040 & 0.012 & -0.048 \\
 & (0.016) & (0.014) & (0.013) & (0.015) \\
+ Linear age & -0.008 & -0.029 & 0.006 & -0.053 \\
 & (0.015) & (0.015) & (0.013) & (0.016) \\
+ Flexible age  & -0.006 & -0.039 & 0.008 & -0.052 \\
 & (0.016) & (0.016) & (0.015) & (0.017) \\
\midrule
Agree & 0.012 & 0.011 & -0.005 & 0.017 \\
 & (0.006) & (0.005) & (0.003) & (0.006) \\
+ Oldest indicator and sex indicator & 0.013 & 0.013 & -0.006 & 0.021 \\
 & (0.007) & (0.006) & (0.003) & (0.006) \\
+ Linear age & 0.006 & 0.010 & -0.004 & 0.024 \\
 & (0.007) & (0.006) & (0.003) & (0.007) \\
+ Flexible age  & 0.008 & 0.013 & -0.004 & 0.024 \\
 & (0.008) & (0.007) & (0.003) & (0.008) \\
\bottomrule
\end{tabular}

    }
    \caption{Robustness Checks}\label{table:robustness}
    \begin{minipage}{\linewidth}
    \footnotesize
    \justifying
    \textit{Note:} 
        Coefficient estimates for the relationship between migration age and outcomes across the four statements (columns (1)–(4)). The table reports results for two outcomes—Average Mean and Agree—and then repeats these outcomes across alternative specifications (rows) that add different controls (adding indicator for the sibling being the oldest, adding linear age, adding flexible age). Each coefficient is followed by its clustered standard error in parentheses.
    \end{minipage}
\end{table}


\begin{table}[H]
    \centering
    \resizebox{\textwidth}{!}{
    \footnotesize
    \begin{tabular}{lcccc}
\toprule \toprule
\textbf{} & \textbf{(1)} & \textbf{(2)} & \textbf{(3)} & \textbf{(4)} \\
\hline
\textbf{} & \makecell{\textbf{Pre-School Suffers} \\ \textbf{if Mother Works}} & \makecell{\textbf{Family Suffers} \\ \textbf{if Mother Works} \\ \textbf{Full Time}} & \makecell{\textbf{Husband and Wife} \\ \textbf{Should Contribute} \\ \textbf{to HH Income}} & \makecell{\textbf{Husband Should Earn,} \\ \textbf{Wife Stay at Home}} \\
\hline
\textbf{Mean Linear Score } \\
\hline
Western Europe, USA, Canada & -0.003 & -0.008 & 0.013 & -0.008 \\
 & (0.020) & (0.024) & (0.022) & (0.030) \\
Rest of the World & -0.027 & -0.028 & 0.014 & -0.053 \\
 & (0.016) & (0.014) & (0.014) & (0.016) \\
\hline
Different from UK & -0.011 & -0.009 & 0.015 & -0.048 \\
 & (0.023) & (0.019) & (0.018) & (0.020) \\
Similar to the UK & -0.003 & -0.011 & -0.011 & 0.004 \\
 & (0.024) & (0.027) & (0.017) & (0.027) \\
\midrule
\textbf{Agree } \\
\hline
Western Europe, USA, Canada & -0.001 & -0.003 & -0.005 & -0.001 \\
 & (0.011) & (0.010) & (0.004) & (0.012) \\
Rest of the World & 0.014 & 0.014 & -0.006 & 0.023 \\
 & (0.006) & (0.006) & (0.004) & (0.007) \\
\hline
Different from UK & 0.014 & 0.016 & -0.003 & 0.024 \\
 & (0.008) & (0.007) & (0.005) & (0.008) \\
Similar to the UK & 0.0115 & 0.0111 & 0.0019 & -0.0049 \\
 & (0.0116) & (0.0114) & (0.0030) & (0.0116) \\
\bottomrule
\end{tabular}

    }
    \caption{Heterogeneity Across Country of Origin}\label{tab:hetro}
\begin{minipage}{\linewidth}
    \footnotesize
    \justifying
    \textit{Note:} 
    Coefficient estimates by subgroup definitions, reported separately for Mean Linear Score and for Agree, across the four statements (columns (1)–(4)). Rows list different subgroup splits (e.g., by region groupings such as Europe vs. Rest of the World, and other region-category definitions). Each cell reports a coefficient estimate with its clustered standard error in parentheses.
    \end{minipage}
\end{table}

\begin{table}[H]
    \centering
    \resizebox{\textwidth}{!}{
    \footnotesize
    \begin{tabular}{@{}l c c c c@{}}
\toprule
  & Preschool child & Family suffers if & Husband and wife & Husband should earn, \\
  &  suffers if mother &  mother works &  should contribute &  wife should \\
  &  works &  full time  &  to HH income  &  stay at home \\
  \hline
Total Variation - Different from UK Vs. Locals & 0.231 & 0.293 & 0.059 & 0.309 \\
 & [0.187, 0.289] & [0.229, 0.353] & [0.03, 0.121] & [0.249, 0.373] \\
 Total Variation - CF Different from UK Vs. Locals & 0.24 & 0.228 & 0.078 & 0.207 \\
 & [0.158, 0.35] & [0.152, 0.334] & [0.042, 0.17] & [0.123, 0.304] \\
Difference & -0.009 & 0.065 & -0.019 & 0.101 \\
 & [-0.102, 0.079] & [-0.027, 0.128] & [-0.085, 0.029] & [0.022, 0.167] \\
\midrule
Total Variation - Similar to the UK Vs. Locals & 0.114 & 0.059 & 0.06 & 0.042 \\
 & [0.072, 0.176] & [0.033, 0.127] & [0.034, 0.121] & [0.028, 0.12] \\
 Total Variation - CF Similar to the UK Vs. Locals & 0.169 & 0.14 & 0.063 & 0.061 \\
 & [0.092, 0.284] & [0.075, 0.259] & [0.036, 0.185] & [0.049, 0.197] \\
Difference & -0.055 & -0.082 & -0.003 & -0.019 \\
 & [-0.163, 0.033] & [-0.183, 0.0] & [-0.102, 0.044] & [-0.126, 0.033] \\
\midrule
MTVD Different from UK Vs. Locals & 0.028 & 0.032 & -0.0 & 0.05 \\
 & [-0.044, 0.071] & [-0.008, 0.072] & [-0.035, 0.036] & [0.01, 0.085] \\
\midrule
MTVD Similar to the UK Vs. Locals & 0.003 & -0.015 & 0.006 & 0.011 \\
 & [-0.061, 0.048] & [-0.068, 0.049] & [-0.036, 0.042] & [-0.037, 0.053] \\
\bottomrule
\end{tabular}

    }
    \caption{Heterogeneity across Country of Origin by Distance from the UK Gender Norms}
    \label{tab:TV_hetro_differentUK}
        \begin{minipage}{\linewidth}
    \footnotesize
    \justifying
    \textit{Note:} This table is analogous to Table \ref{tab:hetro}. In Section 1, we compare the total variation distance between the response distributions of immigrants with from counties who are far from the UK gender norms, as measure in Appendix \ref{appendix:similarity}. Far is considered above the median distance, and the local response distribution for each of the 4 statements.   
    The next row reports the corresponding counterfactual for these immigrants under the scenario in which they were born in the UK. The final row in Section 1 reports the difference between the observed and counterfactual distances. Section 2 presents the same set of quantities for immigrants from countries who are similar to the UK. The last section reports the MTVD for the two groups. Each cell reports a point estimate, with a 95\% confidence interval in square brackets constructed from 1{,}000 bootstrap replications.
    \end{minipage}
    
\end{table}
\section{Additional Figures}
\begin{figure}[H]
    \centering
    \includegraphics[scale=1]{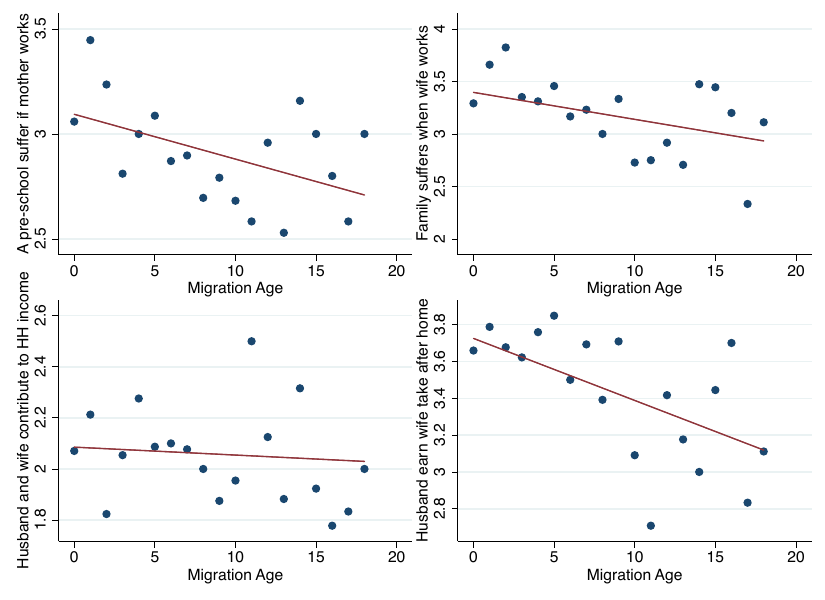}
    \caption{Relation between migration age and attitude towards Gender role for the Siblings sample}
    \label{fig:migrationAge_siblings}
    \begin{minipage}{\linewidth}
    \footnotesize
    \justifying
    \textit{Note:} Four-panel binscatter of migration age (x-axis, years) against the mean coded response (y-axis) for each statement, restricted to the immigrant siblings sample. Each panel overlays a fitted linear trend line on the plotted age-specific mean responses.
    \end{minipage}

\end{figure}

\begin{figure}[H]
    \centering
    \includegraphics[scale=1]{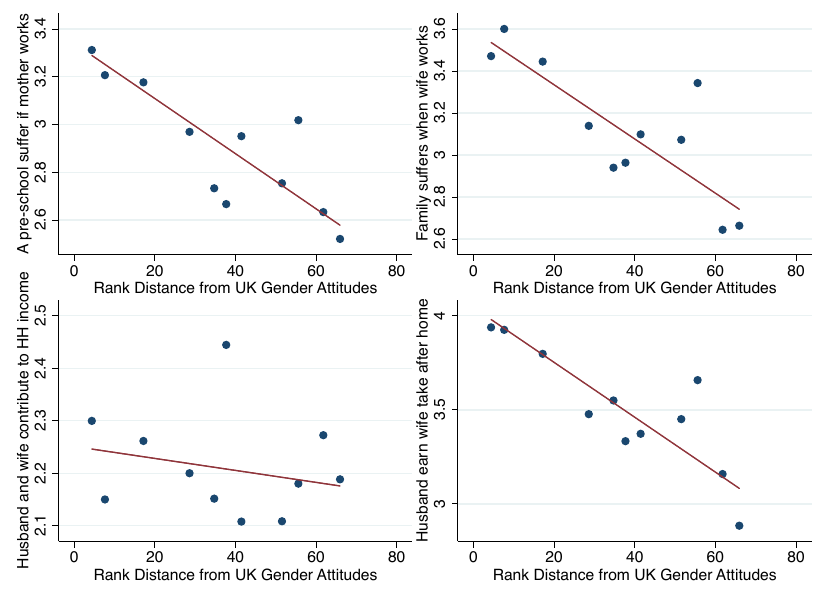}
    \caption{Relation between the origin rank of distance from UK and attitude towards Gender role - All Immigrants}
    \label{fig:distance_allMig}
    \begin{minipage}{\linewidth}
    \footnotesize
    \justifying
    \textit{Note:} Four-panel scatterplot of the origin-country “Rank Distance from UK Gender Attitudes” (x-axis), as describe in Appendix \ref{appendix:similarity}, against the mean coded response (y-axis) for each statement, using all immigrants in the UKHLS data. Each panel overlays a fitted linear trend line on the plotted country-level points.
    \end{minipage}
\end{figure}

\begin{figure}[H]
    \centering
    \includegraphics[scale=1]{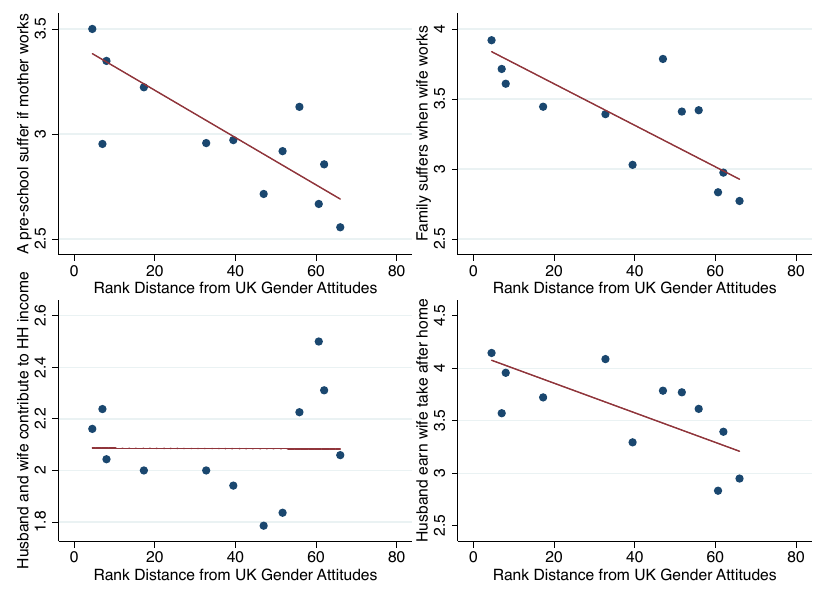}
    \caption{Relation between the origin rank of distance from UK and attitude towards Gender role - Siblings Sample}
    \label{fig:distance_siblings}
    \begin{minipage}{\linewidth}
    \footnotesize
    \justifying
    \textit{Note:} Four-panel scatterplot of the origin-country “Rank Distance from UK Gender Attitudes” (x-axis), as describe in Appendix \ref{appendix:similarity}, against the mean coded response (y-axis) for each statement, using the immigrants siblings sample. Each panel overlays a fitted linear trend line on the plotted country-level points.
    \end{minipage}
\end{figure}


\begin{figure}[H]
    \centering
    \includegraphics[scale=1]{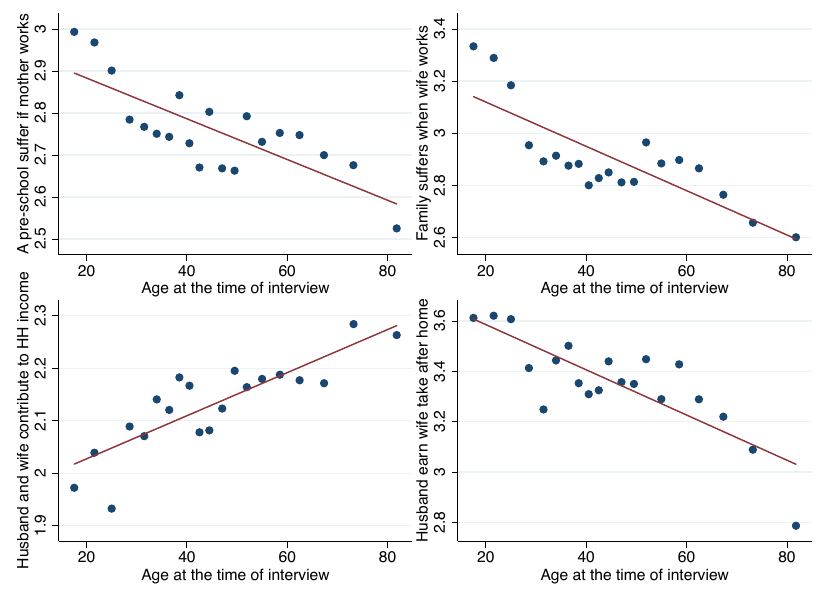}
    \caption{Relation between age at the time of conducting the interview and average linear response - Immigrants}
    \label{fig:immigrants_Age_answers_relation}
    \begin{minipage}{\linewidth}
    \footnotesize
    \justifying
    \textit{Note:} Four-panel scatterplot of respondent age at the time of interview (x-axis) against the mean coded response (y-axis; “average linear response”) for each statement, for immigrants. Each panel overlays a fitted linear trend line on the plotted age-specific mean responses. 
    \end{minipage}
\end{figure}

\begin{figure}[H]
    \centering
    \includegraphics[scale=1]{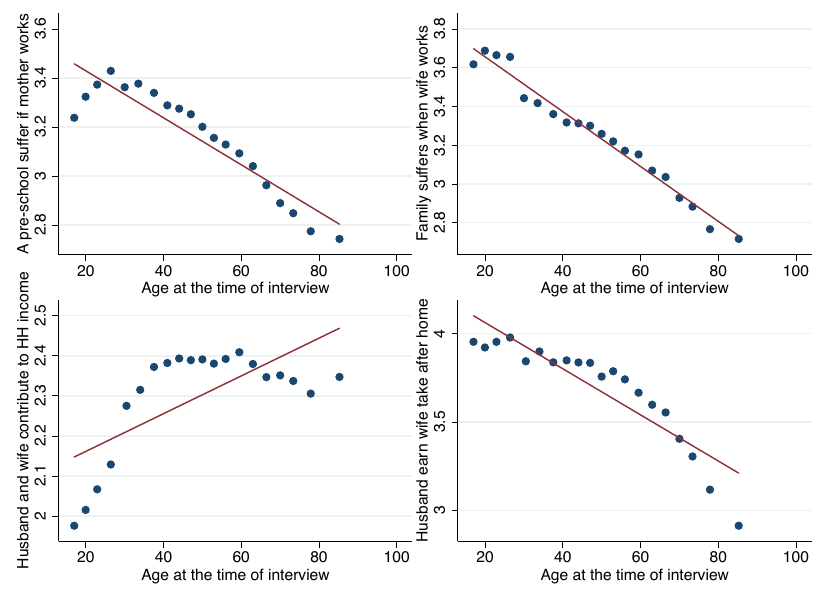}
    \caption{Relation between age at the time of conducting the interview and average linear response - Locals}
    \label{fig:locals_Age_answers_relation}
    \begin{minipage}{\linewidth}
        \footnotesize
        \justifying
        \textit{Note:}  Four-panel scatterplot of respondent age at the time of interview (x-axis) against the mean coded response (y-axis; “average linear response”) for each statement, for locals. Each panel overlays a fitted linear trend line on the plotted age-specific mean responses. The figure does not display standard errors, confidence intervals, or other uncertainty bands. 
        \end{minipage}
\end{figure}

\begin{figure}[H]
    \centering
    \includegraphics[width=0.5\linewidth]{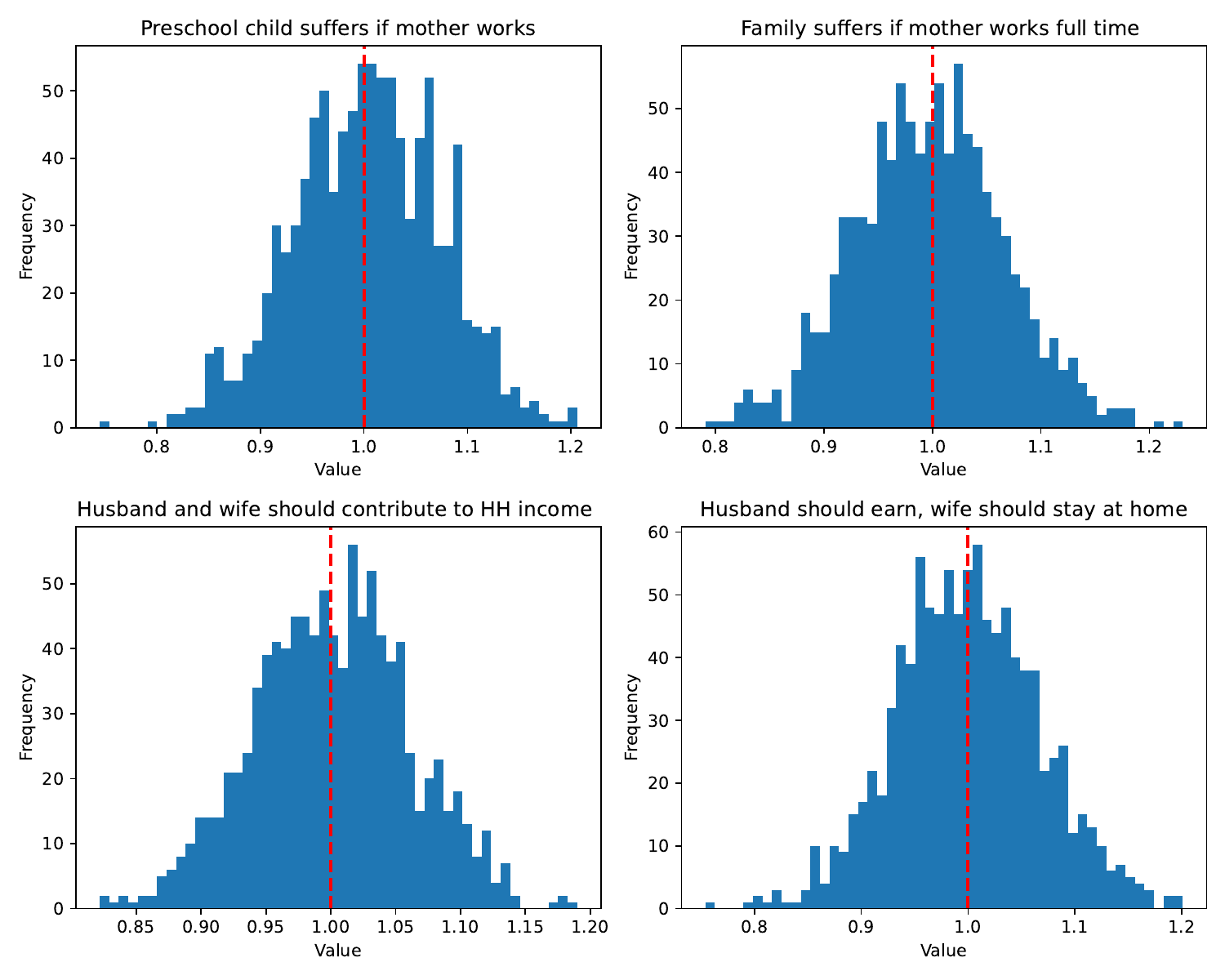}
    \caption{Bootstrap Distribution of Sum of Fixed Effects}
    \label{fig:sumOfFixedEffect}
    \begin{minipage}{\linewidth}
    \footnotesize
    \justifying
    Four-panel histogram of bootstrap draws of the statistic \(\sum_a \gamma^{f(i)}_{q,a}\) (aggregated across individuals \(i\) and questions \(q\)) from Model 3, shown separately for each attitude statement. The x-axis reports the sum value and the y-axis reports the frequency across 1000 bootstrap replications. A dashed vertical reference line marks 1, corresponding to the probability-sum constraint; a solid vertical line marks the bootstrap average as plotted in the figure.  
    \end{minipage}
\end{figure}

\begin{figure}[H]
    \centering
    \includegraphics[width=0.5\linewidth]{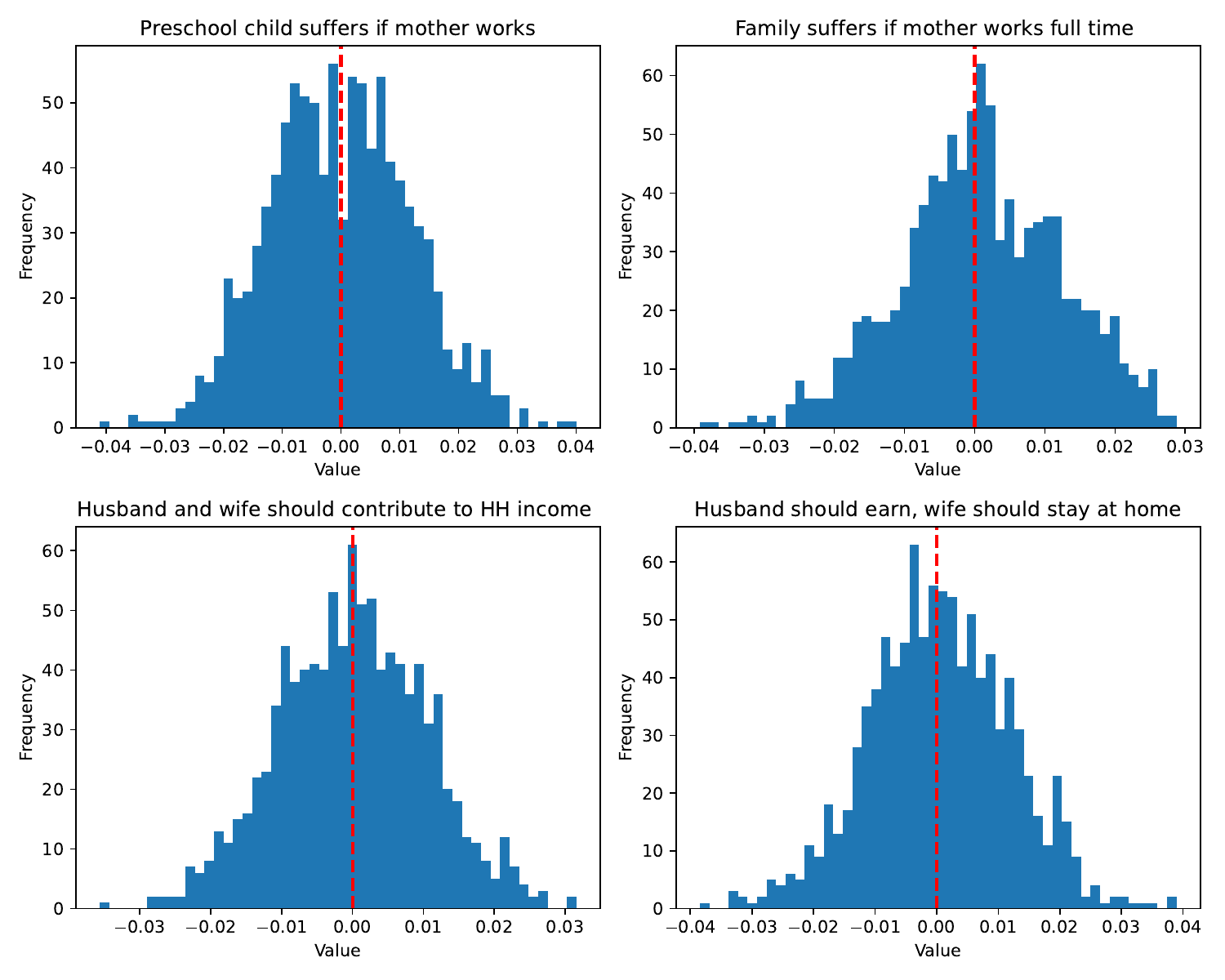}
    \caption{Bootstrap Distribution of Sum of Betas}
    \label{fig:sumOfBetas}
    \begin{minipage}{\linewidth}
    \footnotesize
    \justifying
    Four-panel histogram of bootstrap draws of the statistic \(\sum_a \beta_{q,a}\) from Model 3, shown separately for each attitude statement. The x-axis reports the statistic value and the y-axis reports the frequency across bootstrap replications. A dashed vertical reference line marks 0, corresponding to the constraint that changes in response probabilities across answer categories sum to zero; a solid vertical line marks the bootstrap average as plotted in the figure.  
    \end{minipage}
\end{figure}

\end{document}